\pdfoutput=1
\documentclass[]{amsart}
\usepackage[margin=1.25in]{geometry}
\usepackage{enumitem}

\usepackage{tikz}
\usepackage{tikz-cd}

\usetikzlibrary{decorations.markings,decorations.pathreplacing,
  shapes.geometric,matrix,calc,arrows,chains,positioning,scopes}

\pgfdeclarearrow{
  name = pxto,
  setup code = {
    \pgfarrowssettipend{1.5\pgflinewidth}
    \pgfarrowssetbackend{-2.5508\pgflinewidth}
    \pgfarrowssetlineend{-.25\pgflinewidth}
    \pgfarrowssetvisualbackend{-0.021\pgflinewidth}
    \pgfarrowsupperhullpoint{1.5\pgflinewidth}{0\pgflinewidth}
    \pgfarrowsupperhullpoint{-2.0085\pgflinewidth}{3.6525\pgflinewidth}
    \pgfarrowsupperhullpoint{-2.5508\pgflinewidth}{3.0763\pgflinewidth}
  },
  drawing code = {
    \pgfsetdash{}{0pt}%
    \pgfpathmoveto{\pgfpoint{1.5\pgflinewidth}{0.0254\pgflinewidth}}%
    \pgfpathlineto{\pgfpoint{-2.0085\pgflinewidth}{3.6525\pgflinewidth}}%
    \pgfpathlineto{\pgfpoint{-2.5508\pgflinewidth}{3.0763\pgflinewidth}}%
    \pgfpathlineto{\pgfpoint{-0.4322\pgflinewidth}{0.5\pgflinewidth}}%
    \pgfpathlineto{\pgfpoint{-0.4322\pgflinewidth}{-0.5\pgflinewidth}}%
    \pgfpathlineto{\pgfpoint{-2.5508\pgflinewidth}{-3.0763\pgflinewidth}}%
    \pgfpathlineto{\pgfpoint{-2.0085\pgflinewidth}{-3.6525\pgflinewidth}}%
    \pgfpathclose%
    \pgfusepathqfill
  }
}
\tikzset{>=pxto}
\tikzcdset{arrow style=tikz}
\tikzset{arrow/.style={->}}

\usepackage{quiver}

\usepackage{microtype}
\usepackage[mathscr]{euscript}
\usepackage{mathtools}
\usepackage{thmtools,thm-restate}

\usepackage{upgreek}

\usepackage{todonotes}
\usepackage{comment}

\usepackage{hyperref}

\usepackage[capitalize]{cleveref}

\mathtoolsset{mathic}

\newcommand{\renewtheorem}[1]{%
  \expandafter\let\csname #1\endcsname\relax
  \expandafter\let\csname c@#1\endcsname\relax
  \expandafter\let\csname end#1\endcsname\relax
  \newtheorem{#1}%
}

\theoremstyle{plain}

\renewtheorem{thm}{Theorem}[section]
\crefname{thm}{Theorem}{Theorems}
\renewtheorem{cor}[thm]{Corollary}
\crefname{cor}{Corollary}{Corollaries}
\renewtheorem{lem}[thm]{Lemma}
\crefname{lem}{Lemma}{Lemmas}
\renewtheorem{prop}[thm]{Proposition}
\crefname{prop}{Proposition}{Propositions}

\theoremstyle{definition}

\renewtheorem{rem}[thm]{Remark}
\crefname{rem}{Remark}{Remarks}
\renewtheorem{exa}[thm]{Example}
\crefname{exa}{Example}{Examples}
\renewtheorem{defi}[thm]{Definition}
\crefname{defi}{Definition}{Definitions}

\newcommand*{\constant}[1]{\mathrm{#1}}
\newcommand*{\constructor}[1]{\mathrm{#1}} 

\renewcommand*{\th}{\textsuperscript{th}}
\newcommand*{\from}{:}
\newcommand*{\blank}{\mathord{{-}}}
\newcommand*{\inv}[1]{#1^{-1}}

\newcommand*{\refl}[1]{\mathop{{\constructor{refl}}_{#1}}}

\newcommand*{\inl}{\mathop{\constructor{inl}}}
\newcommand*{\inr}{\mathop{\constructor{inr}}}
\newcommand*{\mrd}{\operatorname{\constructor{mrd}}}
\newcommand*{\glue}{\operatorname{\constructor{glue}}}
\newcommand*{\north}{\constructor{N}}
\newcommand*{\south}{\constructor{S}}
\newcommand*{\trp}[2][]{\mathop{\constant{trp}^{#1}_{#2}}}
\newcommand*{\fst}{\mathop{\constant{fst}}}
\newcommand*{\snd}{\mathop{\constant{snd}}}

\newcommand*{\ap}[1]{\mathop{\left[{#1}\right]}}
\newcommand*{\ptdto}{\to_\ast}%
\newcommand*{\ptdtof}[1]{\overset{#1}{\to}_\ast}%
\newcommand*{\ptdweq}{\weq_\ast}%
\newcommand*{\loopspace}[1]{\mathop{\Omega^{#1}}}%

\newcommand*{\nat}{\mathord{\constant{nat}}}
\newcommand*{\ltr}{\mathord{\constant{ltr}}}
\newcommand*{\rtr}{\mathord{\constant{rtr}}}

\newcommand*{\Id}{\mathord{\constant{Id}}}
\newcommand*{\id}{\mathord{\constant{id}}}
\newcommand*{\cat}[1]{\mathscr{#1}}

\newcommand*{\weq}{\simeq}

\newcommand*{\ZZ}{\mathbb{Z}}
\newcommand*{\NN}{\mathbb{N}}

\newcommand*{\RR}{\mathbb{R}}

\DeclareMathOperator{\Outoperator}{Out}
\newcommand*{\Out}{\Outoperator}
\DeclareMathOperator{\Innoperator}{Inn}
\newcommand*{\Inn}{\Innoperator}
\newcommand*{\inn}{\mathrm{inn}}

\newcommand*{\jdeq}{\equiv}
\newcommand*{\defeq}{\vcentcolon\jdeq}
\newcommand*{\defequi}{\defeq}
\newcommand*{\defis}{\vcentcolon=}
\newcommand*{\leftadjto}{\dashv}

\DeclareMathOperator\pr{pr}
\DeclareMathOperator\im{im}
\DeclareMathOperator\Aut{Aut}

\DeclarePairedDelimiter\Trunc{\lVert}{\rVert}
\DeclarePairedDelimiter\trunc{\lvert}{\rvert} 
\DeclarePairedDelimiter\angled{\langle}{\rangle}
\DeclarePairedDelimiterX\setof[2]\lbrace\rbrace{#1 \mid #2}

\newcommand*{\setTrunc}[1]{\Trunc{#1}_0}
\newcommand*{\settrunc}[1]{\trunc{#1}_0}

\DeclarePairedDelimiterXPP\higherTrunc[2]{}\lVert\rVert{_{#1}}{#2}
\DeclarePairedDelimiterXPP\highertrunc[2]{}\lvert\rvert{_{#1}}{#2}

\newcommand*{\ev}{\constant{ev}}

\newcommand*{\istrunc}[1]{\constant{isTrunc}_{#1}}

\newcommand*{\conncomp}[2]{{#1}_{\left(#2\right)}}

\newcommand*{\UU}{\mathcal{U}}
\newcommand*{\UUp}{\UU_*}
\newcommand*{\UUptd}{\UUp}

\newcommand*{\circled}[1]{\textrm{\textcircled{\small #1}}}

\newcommand*{\hopf}{{\mathcal H}}
\newcommand*{\Sn}[1]{{\mathbb S^{#1}}}
\newcommand*{\Sc}{{\Sn 1}}
\newcommand*{\Sp}{{\Sn 2}}
\newcommand*{\sbt}{\tikz[baseline]{\node[scale=.7,inner
    sep=0, outer sep=0, circle, anchor=base, yshift=.05ex]%
    {$\bullet$};}}%
\newcommand*{\base}{\mathord{\sbt}}
\newcommand*{\Sloop}{\mathord{\circlearrowleft}}
\newcommand*{\bn}[1]{\mathbf{#1}}
\newcommand*{\emptytype}{\emptyset}
\newcommand*{\ind}{\operatorname{\constant{ind}}}
\newcommand*{\susp}{\operatorname{\Sigma}}

\newcommand*{\hgr}[1]{\uppi_{#1}}
\newcommand*{\fgr}{\hgr 1}

\newcommand{\wunit}{0} 
\newcommand{\winv}{{-}}  
\newcommand{\wadd}{{+}}  

\newcommand*{\pathover}[4]{#1 =^{#2}_{#3} #4}

\tikzset{cell/.style={%
    shorten <=1em,%
    shorten >=1em,
    /tikz/commutative diagrams/Rightarrow
  }%
}%
\tikzset{
tikzshortarrow/.style={
    shorten >=0.2cm,
    shorten <=0.2cm,
    thick,
  }
}
\tikzset{pushout/.style={%
    commutative diagrams/.cd,dr,phantom,"\ulcorner", very near end
  }%
}
\tikzset{
    rotninety/.style={anchor=south, rotate=90, inner sep=2pt},
    rottwoseventy/.style={anchor=north, rotate=90, inner sep=2pt},
    rotfortyfive/.style={anchor=south, rotate=45, inner sep=2pt},
    rotthreefifteen/.style={anchor=south, rotate=315, inner sep=2pt},
}

\tikzset{
    ptdto/.style={
      shorten >=.25em,
      to path={-- node[at end,yshift=-3pt] {\scriptsize$\ast$} (\tikztotarget) \tikztonodes},
      }
}

\definecolor{darkgreen}{rgb}{0,0.4,0}
\definecolor{darkblue}{rgb}{0,0,1}
\definecolor{darkred}{rgb}{.7,0,0}

\tikzset{
	tikzforeground/.style={
		->,
		line width = 0.8pt,
		preaction={draw, -, line width=5pt, white},
	},
	tikzbackground/.style={
		->,
		line width = 0.4pt,
	},
	tikzequal/.style={
		-,
		double,
		shorten >=0.2cm,
		shorten <=0.2cm,
		line width = 0.6pt,
		preaction={draw, -, line width=3pt, white},
	},
}

\def\githubpath{\tt\small}

\hypersetup{
    colorlinks,
    linkcolor={red!50!black},
    citecolor={blue!50!black},
    urlcolor={blue!80!black}
}

\setlist[enumerate]{label=(\roman*)}

\begin{document}

\title{On symmetries of spheres in univalent foundations}

\author{Pierre Cagne, Ulrik Buchholtz, Nicolai Kraus, and Marc Bezem}
%
%
%


\begin{abstract}
  Working in univalent foundations,
  we investigate the symmetries of spheres, i.e., the types of the form
  $\Sn n = \Sn n$.
  The case of the circle has a slick answer:
  the symmetries of the circle form two copies of the circle.
  For higher-dimensional spheres, the type of symmetries has again two connected components, namely the components of the maps of degree plus or minus one.
  Each of the two components has $\ZZ/2\ZZ$ as fundamental group.
  For the latter result, we develop an EHP long exact sequence.
\end{abstract}

\maketitle

\section{Introduction}
\label{sec:intro}


Martin-L\"of's dependent type theory~\cite{Martin-Lof-1972} can serve as a basis for proof assistants and dependently typed programming languages.
As pioneered by Voevodsky~\cite{voevodsky_univalentfoundations} as well as by Awodey and Warren~\cite{awodeyWarren_HTmodelsOfIT}, it allows for homotopy-theoretic semantics.
Concretely, types can be interpreted as $\infty$-groupoids~\cite{Kapulkin2021}, and more generally,
as objects in any Grothendieck $(\infty,1)$-topos~\cite{shulman2019}.
These models justify Voevodsky's \emph{univalence axiom} and admit a range of \emph{higher inductive types}.
The field that embraces the view of types as spaces qua homotopy types
is known as \emph{homotopy type theory} (HoTT) and the setting as
\emph{univalent foundations} (UF)~\cite{HoTT}.

It turns out that various results that hold for spaces in standard homotopy theory can be stated and proved for types in homotopy type theory.
The framework enforces all arguments to be purely axiomatic in nature, and the resulting subfield of homotopy type theory is sometimes called \emph{synthetic homotopy theory}.
Examples of results are the type-theoretic Seifert--van Kampen theorem~\cite{houfavonia_et_al:LIPIcs:2016:6562}, the Blakers--Massey connectivity theorem~\cite{favonia_blakers}, and the construction of the Hopf fibration~\cite[Ch.~8.5]{HoTT}.

A central type of study in synthetic homotopy theory is the $n$-dimensional sphere type $\Sn n$.
The calculation of the fundamental group of the circle $\Sn 1$ was among the first results in the area~\cite{licataShulman_circle}, and the result was quickly extended to the $n$\th{} homotopy group of $\Sn n$, i.e., to $\hgr n(\Sn n)$~\cite{licataBrunerie_s1again}.
Further homotopy groups of higher spheres have been studied by Brunerie, 
showing that $\hgr{n+1}(\Sn n) = \ZZ/2 \ZZ$ for $n \geq 3$~\cite{brunerie:thesis}.

In this paper, we are interested in the type of symmetries, 
or self-equivalences, of the spheres $\Sn n$.
By univalence, this type can be written simply as $\Sn n = \Sn n$.
Trivial cases occur for $n = -1$,
where $\Sn{-1}=\emptytype$, the empty type, so $(\Sn{-1} = \Sn{-1}) = \bn 1$,
and for $n= 0$,
where $\Sn 0=\bn 2$, the type of booleans, so $(\Sn 0 = \Sn 0) = \bn 2 = (\bn 1+\bn 1)$.
For $n = 1$, a relatively simple calculation shows that $(\Sn 1 = \Sn 1) = (\Sn 1 + \Sn 1)$.
For $n \geq 2$, a similarly elegant answer does not seem to be possible.
Our main result for this case is that the type $(\Sn n = \Sn n)$ has two equivalent connected components, each with fundamental group $\ZZ / 2 \ZZ$.
Perhaps surprisingly, this turns out to be easier to prove for $n \geq 3$ than for $n = 2$.

Our study is closely related to the calculation of homotopy groups mentioned above.
Recall that the $n$\th{} homotopy group of $\Sn n$ is, by definition, the set-truncation of the iterated loop space $\Omega^n(\Sn n)$.
The latter type is equivalent to $\Sn n \ptdto \Sn n$, 
the type of \emph{pointed} endofunctions on $\Sn n$. In contrast,
we study the types of self-equivalences of $\Sn n$, not the pointed ones.  
Many of our arguments use techniques similar to the ones used in the
calculation of higher homotopy groups, such as the Hopf fibration
\cite[Sec.~8.5]{HoTT}, or Freudenthal's suspension 
theorem~\cite[Thm.~8.6.4]{HoTT}.
For the characterisation of the fundamental group of the components 
of $\Sn 2 = \Sn 2$, Brunerie's~\cite{brunerie:thesis} calculation of 
$\hgr4(\Sn 3)$ is of great use.

The type $\Sn n = \Sn n$ is the type of elements of a (higher) group,
viz., the automorphism group of the $n$-sphere, traditionally 
denoted $\mathrm{G}(n+1)$.
The classifying space $\mathrm{BG}(n+1)$ classifies spherical fibrations with fiber $\Sn n$,
and these play an important role in the branch of geometric topology that deals with the homotopy theory of manifolds via surgery theory (see, e.g., Sullivan~\cite{sullivan1970}).
The homotopy type of the symmetries of the 2-sphere in classical topology has been 
determined by Hansen~\cite{hansen}, and we discuss this in our conclusions.
We refer to Smith~\cite{smith2010} for a general survey of the homotopy theory of 
function types in classical topology.

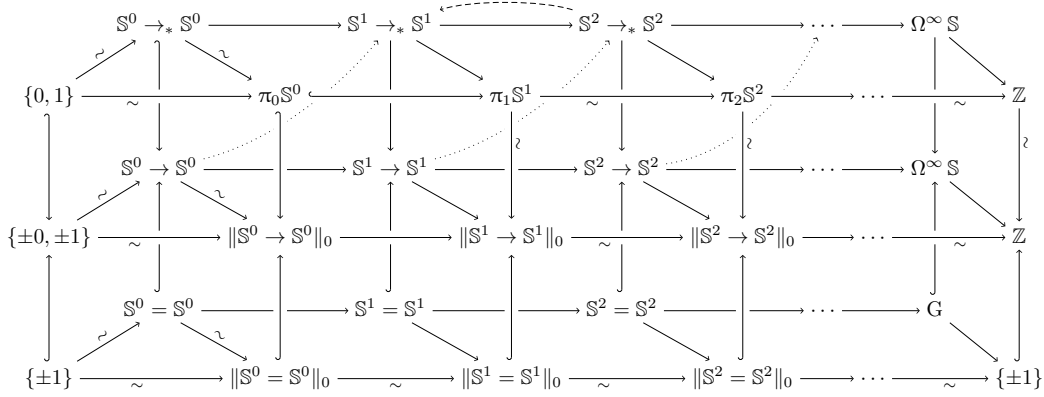
\begin{figure*}
  \caption{Relating pointed endomaps, endomaps, and self-identifications of spheres.}
  \label{fig:comparisons}
\begin{minipage}{\textwidth}
\centering
  \begin{tikzpicture}
  \node[scale=0.8] (a) at (0,0){
  \begin{tikzcd}[cramped, column sep=tiny]
  & \Sn0\ptdto\Sn0 &&[-5pt] \Sc\ptdto\Sc &&[-5pt] \Sp\ptdto\Sp &&[-5pt] \cdots && {\loopspace\infty\mathbb{S}}
  &[5pt] \\
	{\{0,1\}} && {\hgr0\Sn0} && {\fgr\Sc} && {\hgr2\Sp} && \cdots && \ZZ \\
	& \Sn0\to\Sn0 && \Sc\to\Sc && \Sp\to\Sp && \cdots && {\loopspace\infty\mathbb{S}} \\
	{\{\pm0,\pm1\}} && {\setTrunc{\Sn0\to\Sn0}} && {\setTrunc{\Sc\to\Sc}} && {\setTrunc{\Sp\to\Sp}} && \cdots && \ZZ \\
	& {\Sn0=\Sn0} && {\Sc=\Sc} && {\Sp=\Sp} && \cdots && {\mathrm{G}} \\
	{\{\pm1\}} && {\setTrunc{\Sn0=\Sn0}} && {\setTrunc{\Sc=\Sc}} && {\setTrunc{\Sp=\Sp}} && \cdots && {\{\pm1\}}
	\arrow[from=1-2, to=1-4]
	\arrow[from=1-2, to=2-3, "\sim"{rotthreefifteen}]
	\arrow[from=1-2, to=3-2, hook']
	\arrow[from=1-4, to=1-6]
	\arrow[from=1-4, to=2-5]
	\arrow[from=1-4, to=3-4]
	\arrow[from=1-6, to=1-4, curve={height=12pt}, dashed]
	\arrow[from=1-6, to=1-8]
	\arrow[from=1-6, to=2-7]
	\arrow[from=1-6, to=3-6]
	\arrow[from=1-8, to=1-10]
	\arrow[from=1-10, to=2-11]
	\arrow[from=1-10, to=3-10]
	\arrow[from=3-2, to=1-4, curve={height=15pt}, dotted]
	\arrow[from=3-2, to=3-4]
	\arrow[from=3-2, to=4-3, "\sim"{rotthreefifteen}]
	\arrow[from=3-4, to=1-6, curve={height=15pt}, dotted]
	\arrow[from=3-4, to=3-6]
	\arrow[from=3-4, to=4-5]
	\arrow[from=3-6, to=3-8]
	\arrow[from=3-6, to=4-7]
	\arrow[from=3-6, to=1-8, curve={height=22pt}, dotted]
	\arrow[from=3-8, to=3-10]
	\arrow[from=3-10, to=4-11]
	\arrow[from=5-2, to=3-2, hook]
	\arrow[from=5-2, to=5-4]
	\arrow[from=5-2, to=6-3, "\sim"{rotthreefifteen}]
	\arrow[from=5-4, to=3-4, hook]
	\arrow[from=5-4, to=5-6]
	\arrow[from=5-4, to=6-5]
	\arrow[from=5-6, to=3-6, hook]
	\arrow[from=5-6, to=5-8]
	\arrow[from=5-6, to=6-7]
	\arrow[from=5-8, to=5-10]
	\arrow[from=5-10, to=3-10, hook]
	\arrow[from=5-10, to=6-11]
	\arrow[from=2-1, to=1-2, "\sim"{rotfortyfive}]
	\arrow[from=2-1, to=4-1, hook']
	\arrow[from=2-1, to=2-3, crossing over, "\sim"'{pos=.3}]
	\arrow[from=2-3, to=2-5, crossing over, hook]
	\arrow[from=2-3, to=4-3, crossing over, hook']
	\arrow[from=2-5, to=2-7, crossing over, "\sim"'{pos=.3}]
	\arrow[from=2-5, to=4-5, crossing over, "\sim"{pos=0.3,rottwoseventy}]
	\arrow[from=2-7, to=2-9, crossing over]
	\arrow[from=2-7, to=4-7, crossing over, "\sim"{pos=0.3,rottwoseventy}]
	\arrow[from=2-9, to=2-11, crossing over, "\sim"'{pos=.6}]
	\arrow[from=2-11, to=4-11, "\sim"{pos=0.3,rottwoseventy}]
	\arrow[from=4-1, to=3-2, "\sim"{rotfortyfive}]
	\arrow[from=4-1, to=4-3, crossing over, "\sim"'{pos=.3}]
	\arrow[from=4-3, to=4-5, crossing over]
	\arrow[from=4-5, to=4-7, crossing over, "\sim"'{pos=.3}]
	\arrow[from=4-7, to=4-9]
	\arrow[from=4-9, to=4-11, crossing over, "\sim"'{pos=.6}]
	\arrow[from=6-1, to=4-1, hook]
	\arrow[from=6-1, to=5-2, "\sim"{rotfortyfive}]
	\arrow[from=6-1, to=6-3, "\sim"'{pos=.4}]
	\arrow[from=6-3, to=4-3, crossing over, hook]
	\arrow[from=6-3, to=6-5, "\sim"']
	\arrow[from=6-5, to=4-5, crossing over, hook]
	\arrow[from=6-5, to=6-7, "\sim"'{pos=.3}]
	\arrow[from=6-7, to=4-7, crossing over, hook]
	\arrow[from=6-7, to=6-9]
	\arrow[from=6-9, to=6-11, "\sim"'{pos=.6}]
	\arrow[from=6-11, to=4-11, hook]
  \end{tikzcd}
  };
  \end{tikzpicture}
\end{minipage}
\end{figure*}

\paragraph{\bf Setting and assumptions.}
To be precise and to fix notations, we work inside a version of
intuitionistic Martin-L\"of type theory with $\Sigma$-, $\Pi$- and
$\Id$-types and with a cumulative hierarchy of universes,
simply written $\UU$, for which Voevodsky's univalence axiom holds. 
Our type theory corresponds to the one developed in the HoTT book~\cite{HoTT},
although we only assume the higher inductive types specified below.
The basic concepts introduced in the first four chapters of \cite{HoTT}
will mostly be used without further explanation.

We use the same notations as in \cite{HoTT}, with the following exceptions.
If $p:x=y$ and $q:y=z$ are paths, then we denote their
composition by $qp$, or by $q\cdot p$. 
The (dependent) application of a function $f: \prod_{x:X} Y(x)$ 
on paths is denoted by $\ap f$.

We assume a type $\NN$ of natural numbers with its inductive property
(cf.~\cite[Ch.~1.9]{HoTT}), from which is crafted a type $\ZZ$ of integers
as in \cite[\githubpath core/lib/types/Int.agda]{hott-agda}. One important
property of $\ZZ$ is that it has decidable
equality, as proved in the cited file. 

A \emph{pointed type} is a type $A$ with an implicitly or explicitly
given point $a_0 : A$. Given such a pointed type, we write
$\loopspace \null (A) \defequi (a_0 = a_0)$ for the \emph{loop space},
which is itself pointed at $\refl {a_0}$. The
\emph{iterated loop space} is given by $\loopspace 0 (A) \defequi A$
and $\loopspace {n+1} \defequi \loopspace n (\loopspace \null (A))$.
Note that $\loopspace \null$, and hence $\loopspace n$, will be given the
structure of \emph{wild endofunctors} (see \cref{ex:loop-sus-wild-functors}),
which means in particular that they can be applied to functions between 
pointed types. The universe of pointed types is denoted by $\UUptd$, 
and the forgetful map $\UUptd \to \UU$ is
a silent coercion: given a pointed type $A$, its underlying unpointed
type is still written $A$. To avoid confusion, throughout this paper,
the type $A=B$ will always denote the type of paths from $A$ to $B$ in
$\UU$ (that is when $A$ and $B$ are considered as unpointed types),
while $A =_\ast B$ will denote the type of paths from $A$ to $B$ in
$\UUptd$ (that is, when $A$ and $B$ are considered as pointed types).

We further assume the following higher inductive types, referring to
\cite[Ch.~6]{HoTT} for the details: 
the circle (denoted by $\Sc$);
propositional truncation (denoted by $\Trunc A$);
set truncation (denoted by $\setTrunc A$);
suspension (denoted by $\susp A$);
join (denoted by $A*B$);
wedge sum (denoted by $A\vee B$).
The latter three are defined as certain pushouts.
For the circle we assume the usual definition with a base point $\base$
and a loop $\Sloop$, although it can equivalently be described as the
suspension of $\bn 2$.
(In fact, all the above higher inductive types
can be constructed from pushouts alone,
see \cite{Rijke-Join}.)

The spheres $\Sn n$ are then defined by induction on $n:\NN$ by
$\Sn n \defequi \susp(\Sn{n-1})$ for all $n\geq 2$.
We also set $\Sn{-1} \defequi\emptytype$ and 
$\Sn 0\defequi\bn 2 \defequi \bn 1 + \bn 1$.
Then we have $\Sn n = \susp(\Sn{n-1})$ for all $n\ge0$
(for $n=1$ this follows from \cite[Lem.~6.5.1]{HoTT}).
Another basic fact about spheres is that $\Sn n$ is $(n{-}1)$-connected, 
for all $n\geq -1$, \cite[Cor.\ 8.2.2]{HoTT}.
This implies that $\Sc$ is connected and that $\Sp$ is simply connected,
i.e., $\Sp$ and all its path types are connected.

For pointed types we adopt the following conventions.
Suspensions are pointed at $\north$. If $A$ and $B$ are pointed types,
then $A+B$ is pointed at $\inl(a_0)$, with $a_0$ the point of $A$.
A similar convention is followed for pushouts, wedges and joins.

A \emph{fiber sequence} (cf.\ \cite[Def.~8.4.3]{HoTT})
is a sequence $F \overset \iota \ptdto  X \overset p \ptdto Y$
where $F$ is the fiber of $p$ at the point $y_0$ of $Y$
and $\iota$ is the first projection. 
The \emph{connected component} of type $A$ at point $a:A$
is $\conncomp{A}{a}\defequi \sum_{x:A} \Trunc{a=x}$.

Throughout this paper we treat univalence as transparent,
in the sense that equivalences $f:A\weq B$ will be treated as paths 
of type $A=B$ without any warning, and vice versa. (Possible universe 
level issues can be solved by cumulativity and will be disregarded.)
We adopt a similar attitude towards function extensionality.
We treat a homotopy $h:\prod_{x: A} (f(x) = g(x))$ as a path of type $f=g$.
Conversely, any $p: f=g$ is treated as the induced homotopy,
with $p(x)$ (or $p_x$) denoting the induced path of type $f(x) = g(x)$.

\paragraph{\bf Contributions and overview of the paper.}

In \cref{sec:circle-case} we show that the type $\Sc = \Sc$ is
equivalent to $\Sc + \Sc$. To do so, we take a detour to univalent
group theory and establish a far more general result:
Gottlieb's theorem (new in UF). 
We get $(\Sc = \Sc) \weq (\Sc + \Sc)$
by applying Gottlieb's theorem to the group of integers $\ZZ$.
In \cref{sec:sphere} we deal with $\Sp=\Sp$. This case is much more
difficult than the previous. 
Using the Hopf fibration as defined in UF in \cite{brunerie:thesis},
and two definitions of the degree function 
(a variation of \cite{Buchholtz2018CellularCI}), we prove that the 
type $\Sp=\Sp$ has exactly two connected components, 
equivalent to one another.
In \cref{sec:higher-sphere} we prove by induction on
$n\geq 2$ that the type $\Sn n = \Sn n$ has exactly two connected
components, equivalent to one another. Each induction step relies on
Freudenthal's suspension theorem and on the result that suspension
and negation commute (new in UF in 2020).
In \cref{sec:summary} we explain the 3-dimensional \cref{fig:comparisons},
the inventory of the types and maps between them studied so far,
and discuss comparisons with other approaches.
In \cref{sec:whitehead-interlude} we prepare the study of the structure
of the connected components of $\Sp = \Sp$ by some results on the
generalized Whitehead product (new in UF).
We partly develop an EHP long exact sequence (new in UF).
In \cref{sec:structure-components} we show that the fundamental group
of each component of $\Sn n = \Sn n$ is $\ZZ/2\ZZ$ for $n\ge 2$.
Final remarks are made in \cref{sec:conclusions}.
Appendix A contains proofs that are left out or only sketched.
Appendix B provides the basics of
wild categories (including $\UU$ and $\UUptd$), 
wild functors (including $\loopspace \null$ and $\susp$),
wild adjunctions ($\susp\leftadjto\loopspace \null$), and
wild monoids.

\section{Symmetries of the circle}
\label{sec:circle-case}%

In this section, we will prove the following result.
\begin{restatable}{thm}{symmetriesSc}
  \label{thm:symmetries-of-S1}
  There is an equivalence
  \begin{equation*}
    \label{eq:symm-circle}%
    (\Sc = \Sc) \weq (\Sc+\Sc).
  \end{equation*}
\end{restatable}
We will obtain the result as a consequence
of~\cref{thm:gottlieb-univalent}. In order to state the theorem and
prove it, we need a bit of group theory in univalent foundations. For
any group $G$, a delooping of $G$ is a connected pointed 1-type (groupoid) $A$
such that $G = \loopspace\null A$ as groups. Such a delooping always
exist, and two deloopings are always equal as pointed types~\cite{BvDR}. We
usually write $BG$ for such a delooping, with the point
denoted by $s_G$. Given groups $G$ and $H$,
the function $\loopspace\null$ is an equivalence from $BG \ptdto BH$ to
the type of group homomorphisms from $G$ to $H$. Moreover, this
equivalence retricts to an equivalence between $BG \ptdweq BH$ and the
type of group isomorphisms from $G$ to $H$. Recall that for a group
$G$, there is a homomorphism $G \to \Aut(G)$ from $G$ to the group of
group automorphisms of $G$, that sends an element $g:G$ to the
conjugation $x \mapsto g x \inv g$. The kernel of this homomorphism is
called the center of $G$ and is written $Z(G)$. The image of this
homomorphism is called the group of inner automorphisms of $G$ and is
written $\Inn(G)$. The quotient $\Out(G) \defequi \Aut(G)/\Inn(G)$ is
a set whose elements are called the outer automorphisms of $G$. It is
naturally pointed at the class of the identity automorphism.

\begin{thm}\label{thm:gottlieb-univalent}
  Let $G$ be a group. There is a fiber sequence of the form:
  \begin{displaymath}
    BZ(G) \ptdtof{\iota} (BG = BG) \ptdtof{p} \Out (G)
  \end{displaymath}
  where $BG=BG$ is pointed at $\refl {BG}$.
\end{thm}
This theorem is known in classical homotopy theory, and is attributed
to Gottlieb~\cite{gottlieb:subgroup}.

\begin{proof}[Proof of~\cref{thm:gottlieb-univalent}]
  For any type $A$, by pointing $(A=A)$ at $\refl A$ and
  $\setTrunc{A=A}$ at $\settrunc{\refl A}$, there is, by definition of
  the connected component $\conncomp {(A=A)} {\refl A}$, a fiber
  sequence
  \begin{displaymath}
    \conncomp {(A=A)} {\refl A} \ptdtof{\iota} (A=A) 
    \ptdtof{\settrunc{\,}} \setTrunc{A=A}.
  \end{displaymath}
  When $G$ is a group, we can apply this fact to the (unpointed) type
  $BG$. We will get the result stated in~\cref{thm:gottlieb-univalent}
  if we can exhibit pointed equivalences
  $BZ(G) \ptdweq \conncomp {(BG=BG)} {\refl {BG}}$ and
  $\setTrunc{BG=BG} \ptdweq \Out G$. The former is the content
  of~\cref{lem:delooping-center}, and the latter that
  of~\cref{lem:outer-aut-as-components}.
\end{proof}


\begin{restatable}{lem}{lemmadeloopingcenter}
  \label{lem:delooping-center}%
  For any group $G$ there is a pointed equivalence
  $BZ(G) \ptdweq \conncomp {(BG=BG)} {\refl {BG}}$.
\end{restatable}
\begin{proof}[Sketch of proof]
  Define $z_G:\conncomp {(BG=BG)}{\refl {BG}} \ptdto BG$ as the restriction
  of the evaluation $(BG=BG) \to BG$ that maps $x:BG=BG$ to 
  $x(s_G):BG$. The associated group homomorphism
  $\loopspace\null{(z_G)}$ is injective and has image $Z(G)$.
\end{proof}

To prove the next lemma, we need to make observations about subgroups
and quotients in univalent foundations. The curious reader can refer
to \cite[Sec.~5.2 and~5.3]{Symmetry}. Given a group $G$ and a subgroup $H$, the inclusion
$H \subseteq G$ corresponds to a pointed map $i_H : BH \ptdto BG$
(pointed by a path $(i_H)_0$) whose fibers are all sets. 
The fiber $\inv{i_H}(s_G)$ can be
identified with the set $G/H$ of $H$-cosets, in such a way that the
element $(s_H, (i_H)_0)$ corresponds to the class of the neutral element
of $G$. Subsequently, every fiber is merely equivalent to
$G/H$. Moreover, the pointed map $a_H : BG \ptdto \conncomp \UU {G/H}$
mapping $y : BG$ to $\inv{i_H}(y)$ is such that the homorphism
$\loopspace\null(a_H)$ is precisely the action of $G$ on $G/H$ by
multiplication. Now, when $f : G \to G'$ is any group homomorphism,
with corresponding pointed map $Bf : BG \ptdto BG'$, 
we point $\inv{(Bf)}(s_{G'})$ at $(s_G,(Bf)_0)$,
where $(Bf)_0$ is the path pointing $Bf$.
One can then prove that $B\im(f)$ is equivalent to
$\sum_{y':BG'}\setTrunc{\inv{(Bf)}(y')}$, pointed at
$(s_{G'},\settrunc{(s_G,(Bf)_0)})$, 
under which $i_{\im(f)}$ identifies with the first
projection. In particular, there is a pointed equivalence
$G'/\im(f) \ptdweq \setTrunc{\inv{(Bf)}(s_{G'})}$.

\begin{restatable}{lem}{lemmaouterautascomponents}
  \label{lem:outer-aut-as-components}%
  For any group $G$ there is a pointed equivalence
  $\setTrunc{BG = BG} \ptdweq \Out(G)$.
\end{restatable}
\begin{proof}[Sketch of proof]
  Recall that $\Out(G) \jdeq \Aut(G)/\im(\inn)$ for the
  morphism of groups $\inn : G \to \Aut(G)$ that maps an element
  $g\in G$ to the inner automorphism $x \mapsto g x g^{-1}$. Apply the
  observation above to $\inn$ and recognize that the fiber of $B\inn$
  is equivalent to $BG = BG$.
\end{proof}

\cref{thm:gottlieb-univalent} has a simpler statement whenever there
is a section (right inverse) of $p : (BG = BG) \ptdto \Out(G)$.
\begin{cor}
  \label{cor:gottlieb-univalent-section}
  If we have a section $s : \Out(G) \ptdto (BG = BG)$ of the map
  $p : (BG = BG) \ptdto \Out(G)$ in \cref{thm:gottlieb-univalent}, then
  \begin{displaymath}
    (BG=BG) \weq (BZ(G) \times \Out(G)).
  \end{displaymath}
\end{cor}
\begin{proof}
  We know that $(BG=BG) \weq \sum_{\varphi: \Out(G)}\inv p
  (\varphi)$. From~\cref{thm:gottlieb-univalent}, we know that the
  fiber $\inv p (\varphi)$ is equivalent to $BZ(G)$ when $\varphi$ is
  the class of $\id_G : \Aut(G)$. So it suffices to prove that all
  fibers are equivalent to each other. However, since $p$ is a
  set-truncation and we have a section $s$ of it, the fiber
  $\inv p (\varphi)$ is simply the connected component of $BG=BG$ at
  $s(\varphi)$. There is an obvious equivalence from
  $\conncomp {(BG=BG)} {\refl{G}}$ to $\conncomp {(BG=BG)} {s(\varphi)}$,
  namely $\psi \mapsto \psi\circ s(\varphi)$.
\end{proof}

Let us come back to the study of $(\Sc = \Sc)$, and define an element
of it that is not in the connected component of $\refl{\Sc}$. Under
univalence, this is equivalent to defining an equivalence
$\Sc \weq \Sc$ that is not merely equal to $\id_\Sc$. Let
$-\id_\Sc : \Sc \to \Sc$ be the function defined by circle induction
as $-\id_\Sc \defequi \ind(\base,\inv{\Sloop})$. In other words,
$-\id_\Sc$ is the (propositionally) unique function $\Sc \to \Sc$ such
that $-\id_\Sc(\base) \jdeq \base$ and
$\ap{-\id_\Sc}(\Sloop) = \inv{\Sloop}$. It is an equivalence because
it is its own inverse. Indeed, we can construct a proof of
$-\id_\Sc \circ -\id_\Sc = \id_\Sc$ by function extensionality and
$\Sc$-induction: since $\refl\base$ is an element of
$(-\id_\Sc\circ -\id_\Sc)(\base) = \base$, we only need to provide an
element of $\pathover {\refl\base} T \Sloop {\refl\base}$ where $T$ is
the type family
$x \mapsto \left(-\id_\Sc \circ -\id_\Sc \right)(x) = x$. But the
transport in the type family $T$ over $\Sloop$ is given by
$p \mapsto \Sloop \cdot p \cdot \inv{[-\id_\Sc \circ
  -\id_\Sc](\Sloop)}$. Expanding the expression as
$\Sloop \cdot p \cdot \inv \Sloop$, we find that
$\trp[T] \Sloop (\refl\base) = \refl\base$ by simple path algebra, as
we wanted.

Now, we shall prove that:
\begin{equation}
  \label{eq:id-neq-minusid-Sc}
  \id_\Sc \neq -\id_\Sc.
\end{equation}
In order to do so, consider the evaluation fiber sequence:
\begin{equation}
  \label{eq:fiber-sequence-Sc}
  (\Sc \ptdto \Sc) \to (\Sc \to \Sc) \xrightarrow{\ev_{\base}}{} \Sc
\end{equation}
Here, all the fibers can be identified via a function
\begin{equation}
  \label{eq:S1-loopspace-Z-at-each-x}%
  f\from \prod_{x:\Sc} (\base=\base) \weq (x=x)
\end{equation}
with $f(\base)\equiv\id_{\base=\base}$ and $\ap f (\Sloop) \from (\Sloop\blank\inv\Sloop = \id_{\base=\base})$ is
the reflexivity path of $\id_{\base=\base}$ transported using
commutativity in $\base=\base$ and path algebra.
Because $(\base=\base)$ is equivalent
to $\ZZ$, it follows that the sequence gives an equivalence
\begin{equation}
  \label{eq:Sc-to-Sc-equiv-Sc-ZZ}
  (\Sc \to \Sc) \weq \biggl(\sum_{x:\Sc} x = x \biggr)
  \weq \biggl(\sum_{x:\Sc}\ZZ\biggr) \weq \left(\Sc \times \ZZ\right),
\end{equation}
where $\id_\Sc$ is sent to $(\base,1)$ while
$-\id_\Sc$ is sent to $(\base,-1)$. These elements of $\Sc \times \ZZ$
belong to different connected components, so $\id_\Sc$ and $-\id_\Sc$
as well.

We can finally prove \cref{thm:symmetries-of-S1}.
\begin{proof}[Proof of~\cref{thm:symmetries-of-S1}]
  The classifying type $B\ZZ$ of the group $\ZZ$ of integers is
  equivalent to the circle $\Sc$. Because $\ZZ$ is abelian, $BZ(\ZZ)$
  is equivalent to $\Sc$ itself, and $\Inn(\ZZ)$ is trivial. In
  particular, $\Out(\ZZ)$ is the set underlying the group
  $\Aut(\ZZ)$. But $\ZZ$ has exactly two automorphisms, namely the
  identity and $k\mapsto-k$. To apply
  \cref{cor:gottlieb-univalent-section} for the desired result, we
  give a section of the map $p : (\Sc \weq \Sc) \ptdto
  \Out(\ZZ)$. Because we want the section to be pointed, we need to
  send the identity to $\id_\Sc$ and $k\mapsto-k$ to
  $-\id_\Sc$.\footnote{To obtain a section from this reasoning, we
    need first an actual bijection between $\Aut(\ZZ)$ and a set with
    two elements. To construct such a bijection, the decidability of
    equality on $\ZZ$ is crucial: For any automorphism, we need to
    decide if its image on $1$ is $1$ or $-1$.}
\end{proof}

 

Note that~\cref{thm:symmetries-of-S1} also gives that the equivalence~\eqref{eq:Sc-to-Sc-equiv-Sc-ZZ}
restricts to the equivalence $(\Sc\weq\Sc) \weq (\Sc\times\{\pm1\})$
of the corresponding subtypes.
\begin{rem}
  Although we used \cref{thm:gottlieb-univalent} only to prove
  \cref{thm:symmetries-of-S1} here, Gottlieb's result has other
  consequences in univalent group theory, which are worth
  mentioning. For example, we can
  apply~\cref{cor:gottlieb-univalent-section} to the group
  $\mathfrak{S}_n$ of permutations of $n$ elements. By definition, its
  classifying type $B\mathfrak{S}_n$ is equivalent to the connected
  component $\conncomp \UU {\mathbf n}$ of the standard set
  $\mathbf n$ with $n$ elements in the universe $\UU$. A surprising
  fact of group theory is that $\Out(\mathfrak{S}_n)$ is always a
  singleton except for $n=6$ for which it is a $2$-element set. For
  $n\geq 3, n\neq 6$, both the center of $\mathfrak{S}_n$ and the set
  of outer automorphisms are trivial, so we get that
  $\conncomp \UU {\mathbf n} = \conncomp \UU {\mathbf n}$ is
  contractible. For $n=6$, we get that
  $\conncomp \UU {\mathbf 6} = \conncomp \UU {\mathbf 6}$ is a set
  with two elements. In layman's terms, there is an invertible uniform
  way to associate to each $6$-elements set another $6$-elements set,
  and this mapping is drastically different from the identity. This
  mapping can actually be described in more details in terms of graph
  factorizations: for a $6$-element set $X$, consider the complete
  graph on $X$ and then craft the set of sets of perfect matchings not
  sharing any edge; it just happens that this resulting set also has
  $6$-elements.
\end{rem}

\section{Symmetries of the \texorpdfstring{$2$}{2}-sphere}
\label{sec:sphere}

In this section, we will prove that the canonical inclusion
\begin{equation*}
  \conncomp{\left(\Sp = \Sp\right)}{\id_\Sp} +
  \conncomp{\left(\Sp = \Sp\right)}{-\id_\Sp}
  \to
  \left(\Sp = \Sp\right)
\end{equation*}
is an equivalence, i.e., \(\Sp = \Sp\) has exactly two components,
one containing the identity, and one corresponding to the
equivalence $-\id_\Sp \from \Sp \to \Sp$, which is
defined by $\Sp$-induction as the function such that
$-\id_\Sp(\north) \jdeq \south$, and $-\id_\Sp(\south) \jdeq \north$ and
$\ap{-\id_\Sp}(\mrd(x)) = \inv{(\mrd(x))}$ for all $x\from\Sc$.

\begin{restatable}{lem}{lemminusidequivalence}
  \label{lem:minus-id-equivalence}
  The function $-\id_\Sp$ is self-inverse and thus an equivalence.
\end{restatable}
\begin{proof}[Sketch of proof]
  More generally, the same holds for the reflection $-\id_{\susp X} : \susp X\to\susp X$
  on any suspension, and an element of the type
  $\prod_{z:\susp X}\bigl(z = (-\id_{\susp X} \circ -\id_{\susp X})(z)\bigr)$
  is easily constructed by induction.
\end{proof}

The plan now is as follows:
\begin{itemize}
\item First, we give a direct proof that $\id_\Sp$ and $-\id_\Sp$ are not in
  the same connected component;
\item then, we give two definitions of the degree of a self-map $\Sp \to \Sp$,
  from which it follows that every self-equivalence is
  either in the connected component of $\id_\Sp$ or in the connected
  component of $-\id_\Sp$;
\item finally, we prove that the connected components of $\id_\Sp$
  and $-\id_\Sp$ are equivalent to each other.
\end{itemize}
Notice that the last step is less ambitious than in the case of $\Sc$,
where the two connected components were proven equivalent to each
other but also each equivalent to $\Sc$ itself.
We shall see in~\cref{sec:structure-components} that the connected
components of $\id_\Sp$ and $-\id_\Sp$ are not equivalent to $\Sp$
itself.
And indeed, the proof in the case of $\Sc$ relied heavily on two
facts: $\Sc$ is $1$-truncated and $\Sc$ is the classifying type of an
abelian group. In other words, the homotopy structure of $\Sc$ is very
well understood. This is not the case for $\Sp$: for example, it is
certainly not $2$-truncated (\cite{brunerie:thesis}), and is expected
to be provably not $n$-truncated for any $n$.

The main tool for this section is the Hopf family, as defined by Brunerie
in~\cite{brunerie:thesis}, to get an analogue in HoTT of the Hopf
fibration in topology. We define, uniformly in $x\from\Sc$, the
function $\iota_x\from \Sc\to\Sc$ by $\Sc$-induction
(giving the usual H-space structure on $\Sc$), putting
$\iota_x(\base) \jdeq x$ and $\ap{\iota_x}(\Sloop) = f_x(\Sloop)$.
Here, $f\from \prod_{x\from\Sc} (\base=\base) \weq (x=x)$ is the
dependent function defined in~\eqref{eq:S1-loopspace-Z-at-each-x}.
Clearly, $\iota_{\base}=\id_{\Sc}$ and hence, since $\Sc$ is connected,
every $\iota_x$ is merely equal to $\id_{\Sc}$ and thus an equivalence.
Recalling the transparency of univalence,
we view $\iota_x\from \Sc=\Sc$ as a path.
Note also that $\iota_x$ is the element
of $\conncomp{(\Sc=\Sc)}{\id_\Sc}$ that corresponds to $x:\Sc$
under the evaluation equivalence $\conncomp{(\Sc=\Sc)}{\id_\Sc} \weq \Sc$
exhibited in~\cref{sec:circle-case}. Now define the type family
$\hopf\from \Sp \to \UU$ by $\Sp$-induction as the family:
\begin{displaymath}
  \hopf(\north) \jdeq \Sc,\quad
  \hopf(\south) \jdeq \Sc,\quad\text{and}\quad
  \ap\hopf (\mrd(x)) = \iota_x\ \text{for all}\ x:\Sc
\end{displaymath}

Following Brunerie's exposition, we consider the map
\begin{equation}\label{eq:tau}
  \tau: \loopspace\null\Sp \ptdto \Sc, \, p \mapsto \ap\hopf(p)(\base),
\text{ pointed by $\refl{\base}:\tau(\refl\north)=\base$}.
\end{equation}
This map is the key to getting the second homotopy group of
the sphere (cf.\ \cite[Sec.~8.4 and 8.5]{HoTT}). For now, recall from
\cref{ex:loop-sus-wild-functors} that there is a pointed map
\begin{displaymath}
  \eta_\Sc : \Sc \ptdto \loopspace\null \Sp,\quad
  x\mapsto \inv{\mrd(\base)}\cdot \mrd(x)
\end{displaymath}
where the pointing path $\eta_0 : \eta_\Sc(\base) = \refl N$ is given by path algebra.
\begin{restatable}{lem}{lemtauretractionunitcircle}
  \label{lem:tau-retraction-unit-circle}
  The map $\tau$ is a retraction of $\eta_\Sc$, meaning that there is an element of $\tau\circ\eta_\Sc
  = \id_\Sc$ as pointed functions.
\end{restatable}
We remark that this is also an instance of a general fact about left-invertible
H-spaces~\cite[Prop.~2.19]{BCFR}.

The following is proved by circle induction.
\begin{restatable}{lem}{leminvetascequalsetascminusid}
  \label{lemma:inv-eta-Sc-equals-eta-Sc-minusid}
  There is an element of the type $\inv{\eta_\Sc(\blank)} = \eta_\Sc
  \circ {-\id_\Sc}$.
\end{restatable}

\begin{lem}
  \label{lemma:S2-id-neq-minusid}%
  The proposition $\id_\Sp \neq -\id_\Sp$ holds.
\end{lem}
\begin{proof}
  Suppose $p:\id_\Sp = -\id_\Sp$ and derive a contradiction. Through
  function extensionality, it produces paths
  \begin{displaymath}
    p(\north) \from \north = \south
    \quad\text{and}\quad
    p(\south) \from \south = \north,
  \end{displaymath}
  and for all $x\from\Sc$ a path over,
  $\ap p (\mrd(x))\from \pathover {p(\north)} T {\mrd(x)} {p(\south)}$, where
  $T\from \Sp \to \UU$ is the type family
  $T(a)\defequi (\id_\Sp(a) = -\id_\Sp(a))$. Because $\Sp$ is simply
  connected and we are targeting the empty type $\varnothing$, which
  is a proposition, we might as well assume paths of types $p(\north)=\mrd(\base)$
  and $p(\south) = \inv{\mrd(\base)}$. Transporting $\ap p(\mrd(x))$ over
  these two paths, we get a path of type
  $\pathover {\mrd(\base)} T {\mrd(x)} {\inv{\mrd(\base)}}$.
  Transport over $\mrd(x)$ in the type family $T$ is
  the function
  $q \mapsto \inv{\mrd(x)} q \inv{\mrd(x)}$, so we get a path
  \begin{displaymath}
    \inv{\mrd(x)} \mrd(\base) \inv{\mrd(x)} = \inv{\mrd(\base)}
    \quad
    \text{for all }x\from\Sc
  \end{displaymath}
  Equivalently, this is a path $\inv{\eta_\Sc(\blank)} = \eta_{\Sc}$. Compose
  with the path from \cref{lemma:inv-eta-Sc-equals-eta-Sc-minusid} to get a
  path $\eta_\Sc\circ {-\id_\Sc} = \eta_\Sc$. Using
  \cref{lem:tau-retraction-unit-circle}, we conclude that $-\id_\Sc =
  \tau\circ\eta_\Sc\circ {-\id_\Sc} = \tau\circ\eta_\Sc = \id_\Sc$, which we
  already know to be absurd.
\end{proof}

This proves that $\id_\Sp$ and $-\id_\Sp$ belong
to different connected components. We proceed to the second step
of the road map: every equivalence in $\Sp \simeq \Sp$ is
either in the component of $\id_\Sp$ or in the component of $-\id_\Sp$.
To this end we construct two, ultimately equal, degree functions
$d,d':(\Sp\to\Sp)\to\ZZ$:
\begin{enumerate}
\item The first, $d$, is directly seen to be a
    morphism of wild monoids,
    where the operations are given by composition and multiplication, respectively.
    In particular, it maps equivalences to
    invertible elements in $\ZZ$, that is $1$ or $-1$.
\item The second, $d'$, is more easily seen to be `weakly injective', i.e., $d(f)=d(g)$
  implies $\Trunc{f=g}$. From~\cref{lemma:S2-id-neq-minusid} we then get
  that the degree of $-\id_\Sp$ is $-1$, and the degree
  induces equivalences $\Trunc{\Sp\to\Sp}_0 \weq \ZZ$ and
  $\Trunc{\Sp=\Sp}_0 \weq \{\pm1\}$. 
\end{enumerate}

To define the degree, we recall
from \cite[Cor.~8.5.2]{HoTT} that the second homotopy group of $\Sp$ is $\ZZ$.
Indeed, the second homotopy group $\hgr 2(\Sp)$ is defined as the
set-truncation $\setTrunc{\loopspace 2 \Sp}$, and \cite[Sect.~8.4 and 8.5]{HoTT} proves that the map $\loopspace\null\tau :
\loopspace 2 \Sp \to \loopspace \null \Sc$ induces an isomorphism
$\setTrunc{\loopspace\null\tau}: \hgr 2(\Sp) \to \hgr 1(\Sc)$ on the set-truncations,
with $\setTrunc{\loopspace\null\eta}$ as the inverse.
Indeed, by the Freudenthal suspension theorem~\cite[Thm.~8.6.4]{HoTT},
$\eta$ is $0$-connected, but it has a retraction, so it (and $\tau$) are $1$-equivalences.
(A $1$-equivalence is a map that induces an equivalence on $1$-truncations.)
Composing with the isomorphism $\hgr 1(\Sc) \weq \ZZ$ \cite[Cor.~8.1.11]{HoTT}, we
get a group isomorphism $\zeta: \hgr 2(\Sp) \to \ZZ$. Our first definition
of the degree of a
pointed map is then as the image of $1$ through the induced homomorphism
on second homotopy groups, transported back and forth by $\zeta$:
\begin{displaymath}
  d(f) \defequi \left(\zeta\circ\hgr 2(f)\circ \inv\zeta\right)(1):\ZZ \qquad
  \text{for any}\ f:\Sp\ptdto\Sp.
\end{displaymath}

Note that the degree of a map is \emph{a~priori} defined only when the map is pointed.
However, the following lemma shows that the degree is independent of the choice of 
such a path.
\begin{restatable}{lem}{lemmaforgetpointsconn} \label{lem:forget-points-conn}
  \label{lem:deg-independent-path}
    Let $X$ and $Y$ be types, and $x : X$ and $y : Y$ be points. If $Y$ is $n$-connected, 
    then the map which forgets the pointing paths,
    $\pr_1:  ((X,x) \ptdto (Y,y)) \to (X \to Y)$, 
    is $(n{-}1)$-connected.
\end{restatable}
\begin{proof}
  Let $f: X\to Y$. The fiber ${\inv \pr_1 (f)}$ is equivalent to $f(x)=y$.
  Now apply \cite[Thm.~7.3.12]{HoTT} and use the assumption. 
\end{proof}
In particular, if we have two paths $f_0,f_0': f(\north) = \north$
pointing the map $f: \Sp \to \Sp$,
then the proposition $d(f,f_0) = d(f,f'_0)$ holds,
since $d(f,\blank) : (f(\north) = \north) \to \ZZ$ is a map
from a connected type to a set, hence constant.

\begin{restatable}{prop}{propdegreemonoidmorphism}
  \label{prop:degree-monoid-morphism}
  The degree function $d$ is a morphism of wild monoids from $\Sp \to \Sp$
  (\cref{rem:wild-monoids})
  to the multiplicative monoid $\ZZ$.
  \label{prop:deg-monoid-morphism}
\end{restatable}
\begin{proof}[Sketch of proof]
  Clearly, $d(\id_\Sp)=1$.
  From~\cref{lem:deg-independent-path} we may assume $f,g:\Sp\to\Sp$ are pointed.
  By functoriality of the second fundamental group we have
  $\hgr 2(g\circ f) = \hgr 2(g)\circ\hgr 2(f)$,
  from which we readily conclude $d(g\circ f)=d(g)d(f)$.
\end{proof}

\begin{rem}
  As in~\cite{Buchholtz2018CellularCI}, one can also put a group structure
  on $\Sp \ptdto \Sp$ such that the degree function becomes a
  group morphism onto $\ZZ$ with its additive structure.
  Together with \cref{prop:degree-monoid-morphism},
  the degree function on pointed maps then becomes a wild ring morphism.
  In general, pointed self-maps $\susp X \ptdto \susp X$ only form a
  wild \emph{near-ring}.
\end{rem}

\begin{restatable}{cor}{cordegreeequivalences}
  \label{cor:degree-equivalences}
  The degree of a self-equivalence of the sphere $\Sp$ is either $1$ or $-1$.
\end{restatable}
\begin{proof}[Sketch of proof]
  For a self-equivalence $f:\Sp\weq\Sp$
  we have $1=d(\id_\Sp)=d(f\circ f^{-1}) = d(f)d(f^{-1})$,
  so $d(f)$ and $d(f^{-1})$ are multiplicatively inverse integers, hence $\pm1$.
\end{proof}

To prove that the degree map is an injection on connected components, we will define
another map $\bar d: (\Sp \ptdto \Sp) \to \ZZ$, which is easily proven an
injection on connected components, and then we will prove that $d = \bar d$.

Recall from \cref{ex:loop-sus-wild-functors} that for each pointed type $A$,
there is a map $\eta_A : A \ptdto \loopspace\null \susp A$. These maps are such
that the following function is an equivalence for each $A,B:\UUptd$ (see
\cref{def:wild-adj}):
\begin{equation}
  \Phi_{A,B} \defequi \loopspace \null \blank \circ \eta_A :  (\susp A \ptdto B) \weq (A \ptdto \loopspace \null B)
  \label{eq:triangle-identity-unit}
\end{equation}
There is now an equivalence $\gamma:(\Sp \ptdto \Sp) \weq \loopspace 2 \Sp$ defined as
the composition:
\begin{displaymath}
  (\Sp \ptdto \Sp) \stackrel{\Phi_{\Sc,\Sp}}\weq
  (\Sc \ptdto\loopspace\null \Sp) \stackrel{\Phi_{\bn 2,\loopspace\null\Sp}}\weq
  (\bn 2\ptdto\loopspace 2 \Sp) \stackrel{\ev_1}\weq
  \loopspace 2 \Sp,
\end{displaymath}
using that $\Sc \weq \susp{\bn 2}$, and with $\ev_1$ evaluating its argument
at the non-base point of $\bn 2$.
We now define $\bar d$ as the composition
\begin{displaymath}
  (\Sp \ptdto \Sp) \stackrel \gamma \to \loopspace 2 \Sp \stackrel {\settrunc\blank} \to
  \hgr 2(\Sp) \stackrel \zeta \to \ZZ.
\end{displaymath}


We now show that $d$ and $\bar d$ coincide, allowing us to compute the degree
through $\bar d$ when necessary.
\begin{prop}
  The equation $d=\bar d$ holds.
  \label{prop:alternative-description-degree}
\end{prop}
\begin{proof}
  Taking the definitions of $d$ and $\bar d$ into account,
  we have to prove $\hgr 2 (f)({\inv\zeta(1)})= \settrunc {\gamma(f)}$ for all $f$.
  Unfolding definitions,
  we have to prove that the outer diagram commutes in the following:
  \begin{equation*}
    \begin{tikzpicture}[x=3cm,y=-.92cm]
	\tikzstyle{ar} = [draw, ->]
	\tikzstyle{round} = [circle,draw,inner sep=0.1em],

	\node (A1) at (0,0) {$(\Sp \ptdto \Sp)$};
	\node (A2) at (0,1) {$(\loopspace \null \Sp \ptdto \loopspace\null\Sp)$};
	\node (A3) at (0,2) {$(\loopspace 2 \Sp \ptdto \loopspace 2 \Sp)$};	
	\node (A4) at (0,3) {$(\hgr 2(\Sp) \to \hgr 2(\Sp))$};	

	\node (B1) at (1,1.2) {$(\loopspace\null\Sc \ptdto \loopspace 2 \Sp)$};
	\node (B2) at (1,4) {$\hgr 2(\Sp)$};

	\node (C1) at (2,0) {$(\Sc \ptdto \loopspace\null\Sp)$};
	\node (C2) at (2,2) {$(\bn 2 \ptdto \loopspace 2 \Sp)$};
	\node (C3) at (2,3) {$\loopspace 2 \Sp$};

	\small 
	\path[ar] (A1) -- node[left] {$\Omega$} (A2);
	\path[ar] (A2) -- node[left] {$\Omega$} (A3);
	\path[ar] (A3) -- node[left] {$\setTrunc\blank$} (A4);
	\path[ar] (A1) -- node[above] {$\Phi_{\Sc,\Sp}$} (C1);
	\path[ar] (A2) -- node[below] {$\blank\circ\eta_{\Sc}$} (C1);
	\path[ar] (C1) -- node[below right] {$\Omega$} (B1);
	\path[ar] (C1) -- node[right] {$\Phi_{\bn 2,\loopspace\null\Sp}$} (C2);
	\path[ar] (C2) -- node[right] {$\ev_1$} (C3);
	\path[ar] (A4) -- node[below left] {$\ev_{\inv\zeta(1)}$} (B2);
	\path[ar] (C3) -- node[below right] {$\settrunc\blank$} (B2);
	\path[ar] (A3) -- node[below right] {$\blank\circ \loopspace\null(\eta_{\Sc})$} (B1);
	\path[ar] (B1) -- node[below left] {$\blank\circ\eta_{\bn 2}$} (C2);
	\path[ar] (A3) -- node[below left] {$\ev_\ell$} (C3);

	\node[round] (X1) at (0.5,0.35) {$1$};
	\node[round] (X2) at (0.35,1.45) {$2$};
	\node[round] (X3) at (1.8,1.15) {$3$};
	\node[round] (X4) at (1.3,2.2) {$4$};
	\node[round] (X5) at (.9,3.2) {$5$};

%
    \end{tikzpicture}
  \end{equation*}
  We will do that by proving that each of the small inner diagrams
  denoted $\circled{1},\dots,\circled{5}$ commute, for elementary
  reasons. Triangles $\circled{1}$ and $\circled{3}$ commute as
  instances of \eqref{eq:triangle-identity-unit}. The commutativity of
  square $\circled{2}$ simply expresses the functoriality of
  $\loopspace\null$. In order to prove that $\circled{4}$ and
  $\circled{5}$ commute, we first need to define $\ell$ in $\ev_\ell$:
  we set $\ell \defequi \loopspace\null (\eta_\Sc)(\Sloop)$.  
  Then it is almost immediate that $\circled{4}$ commutes because under the
  equivalence $\Sc = \susp {\bn 2}$, we have
  $\Sloop = \eta_{\bn 2}(1)$.  Now the commutativity of $\circled{5}$
  will follow from the functoriality of $\setTrunc\blank$ once we have
  shown $\settrunc \ell = \inv\zeta(1)$. By definition of $\zeta$,
  this is equivalent to showing that
  $\loopspace\null(\tau) (\ell) = \Sloop$ in $\loopspace\null
  \Sc$. However, we have seen in \cref{lem:tau-retraction-unit-circle} that
  $\tau$ is a retraction of $\eta_{\Sc}$. Hence it follows that:
  \begin{displaymath}
    \loopspace\null(\tau)(\ell) = \loopspace\null(\tau)(\loopspace\null(\eta_\Sc)(\Sloop)) = \Sloop.\qedhere
  \end{displaymath}
\end{proof}

\begin{cor}
  \label{prop:symm-S2-connected-components}
  \label{cor:equivalence-conn-component}
  The degree map $d:(\Sp \to \Sp) \to \ZZ$ is $0$-connected. Hence any
  self-equivalence of\/ $\Sp$ is in the connected component of
  either $\id_\Sp$ or $-\id_\Sp$, and the canonical inclusion
  \begin{displaymath}
    \conncomp{(\Sp = \Sp)}{\id_\Sp} + \conncomp{(\Sp = \Sp)}{-\id_\Sp} \to (\Sp = \Sp)
  \end{displaymath}
  is an equivalence.
\end{cor}
\begin{proof}
  The previous result shows that $d(f) = \bar d(f) \jdeq
  \zeta(\settrunc{\gamma(f)})$ with $\gamma, \zeta$ equivalences. As
  $\settrunc\blank$ is $0$-connected, so is $d$.
  From~\cref{cor:degree-equivalences} we know any self-equivalence has degree $\pm1$,
  and since $-\id_\Sp$ is in a different component than $\id_\Sp$
  by~\cref{lemma:S2-id-neq-minusid}, we get that $d(-\id_\Sp)=-1$.
\end{proof}

Next, we show that the two components
$\conncomp{(\Sp = \Sp)}{\id_\Sp}$ and
$\conncomp{(\Sp = \Sp)}{-\id_\Sp}$ are equivalent.
In fact, this has little to do with $\Sp$ itself, and one can state a more
general result.
\begin{prop} \label{prop:equiv-susp-comp}
  Let $A$ be a type with a point $a:A$ and a loop $p : a = a$. Then:
  \begin{displaymath}
    \conncomp{(a = a)}{\refl{a}} \weq \conncomp{(a = a)}{p}
  \end{displaymath}
\end{prop}
\begin{proof}
  Define $f:{(a = a)} \to (a = a)$ by mapping
  $q$ to $qp$.
  Then $f$ is an equivalence with pseudo-inverse given by mapping $r$ to $r\inv p$.
  Moreover, we have $f(\refl a) = p$.

  The equivalence $f$ then restricts to an equivalence between the connected
  component of $\refl a$ and the connected component of $p$.
\end{proof}
\begin{rem}
  Another way to state this result is that for any $\infty$-group $G$
  with an element $g$, we get an equivalence $\conncomp Ge \weq \conncomp Gg$
  by mapping $h$ to $hg$.
  Indeed, the point $a$ in $A$ has an $\infty$-group of symmetries,
  whose elements form the type $a=a$.

  Of course, a similar result holds generally for wild groups.
\end{rem}


We shall show in~\ref{sec:structure-components} that $\fgr(\Sp=\Sp,\id_\Sp)\weq\ZZ/2\ZZ$.
From~\cref{prop:equiv-susp-comp} we then also get $\fgr(\Sp=\Sp,-\id_\Sp)\weq\ZZ/2\ZZ$.

\section{Symmetries of the \texorpdfstring{$n$}{n}-sphere}
\label{sec:higher-sphere}


Having discussed the cases $n = 1$ and $n = 2$ in some detail, we finally wish to establish the result that $\Sn n = \Sn n$ has two connected components, for all $n \geq 1$.
The main ideas for the proof are already contained in the result that $\hgr n(\Sn n) = \ZZ$ \cite{licataBrunerie_s1again} and the definition of degrees by Buchholtz and Favonia \cite{Buchholtz2018CellularCI}, but the argument that we need does not seem to have been written out
informally in detail so far, though it has been formalized.

An important tool in this section is the wild adjunction 
$\susp \dashv \Omega$ from \cref{prop:susp-loop-adjunction}. 
We recall that the suspension $\susp$ acts like a wild functor, 
see \cref{ex:loop-sus-wild-functors}(ii), on morphisms:
For $f : A \to B$, define $\susp(f) : \susp(A) \to \susp(B)$ by 
$\north_{\susp(A)} \mapsto \north_{\susp(B)}$,
$\south_{\susp(A)} \mapsto \south_{\susp(B)}$,
and $\mrd_{\susp(A)}(a) \mapsto \mrd_{\susp(B)}(f(a))$.
The map $\susp(f)$ is pointed by the reflexivity path.

A core result is the following, to be proved in
\cref{sec:susp-is-monoid-iso}:
\begin{restatable}{thm}{suspiso}
    \label{thm:susp-is-monoid-iso}
    For all natural numbers $n \geq 1$, the wild monoid morphism,
    \begin{equation*}
    (\Sn n \to_* \Sn n) \xrightarrow{\susp} (\Sn {n+1} \to_* \Sn {n+1})
    \end{equation*}
    is a $0$-equivalence.
\end{restatable}

Recall that an $n$-equivalence is a map that becomes an 
equivalence after $n$-truncation, cf.~\cite[Sec. 2]{CORS}.
In \cref{sec:conncomps-Sn=Sn} we shall see how
\cref{thm:susp-is-monoid-iso} implies that $\susp$ induces an isomorphism
$\setTrunc {\Sn n \to \Sn n} \to \setTrunc {\Sn {n+1} \to \Sn {n+1}}$
of monoids. Hence, by induction, $\setTrunc {\Sn n \to \Sn n}$
and $(\ZZ, \times)$ are isomorphic as monoids, and so
$\Sn n = \Sn n$ has two connected components.
In \cref{sec:-id<>id-Sn=Sn} we give one concrete symmetry
in each of the components.

\subsection{The suspension morphism is $0$-connected}
\label{sec:susp-is-monoid-iso}
Recall the equivalence $\Phi_{A,B}$ and the unit $\eta_A$ in
 \cref{sec:sphere}, \eqref{eq:triangle-identity-unit}
 and \cref{prop:susp-loop-adjunction}.
Taking $\susp B$ for $B$ in \eqref{eq:triangle-identity-unit} we get:
\[
  \Phi_{A,\susp B} \defequi \loopspace \null \blank \circ \eta_A :
   (\susp A \ptdto \susp B) \weq (A \ptdto \loopspace \null \susp B)
\]
We will use the following naturality properties.
\begin{lem} \label{lem:adj-prop}
    For any pointed map $f:A\ptdto B$ we have:
\[
(\Phi_{A,\susp B} \circ \susp)(f) \jdeq
 \loopspace \null (\susp(f)) \circ \eta_A = \eta_B \circ f
\]
\end{lem}
\begin{proof}
  By the naturality witness $\nat_\eta$ from \cref{def:wild-adj}.
\end{proof}
  


\begin{lem} \label{lem:ap-Sigma}
    Let $X$ be a pointed type and $f : A \ptdto B$ a pointed function.
    Then for any $g:\susp X \ptdto A$ we have:
\[
   \Phi_{X,B}(f\circ g) = \loopspace\null (f) \circ \Phi_{X,A}(g)
\]
\end{lem}
\begin{proof}
    This is the last part of \cref{def:wild-adj}.
    %
\end{proof}

\begin{defi}\label{def:sphere-to-eq-loops}
Let $A$ be a pointed type.
Recall that $\Sn 0 \defequi \bn 2$ is pointed at $\inl({\ast})$
and that $(\Sn0 \ptdto A)\simeq A$ by the equivalence 
$\phi_A^0(f:\Sn0 \ptdto A) \defequi f(\inr(\ast))$. 
For $n\geq 1$, define equivalences 
$\phi_A^n: (\susp \Sn{n-1} \ptdto A) \to\loopspace n A$ 
by induction:
$\phi_A^n \defequi \phi_{\loopspace\null A}^{n-1} \circ \Phi_{\Sn {n-1},A}$.
\end{defi}

\begin{lem}\label{lem:iterated-ap-Sigma}
    For all $n \geq 0$ and $f : A \ptdto B$ and $g:\Sn n \ptdto A$ 
    we have: 
\[
\phi_B^n (f\circ g) = \Omega^n(f)(\phi_A^n(g))
\]
\end{lem}
\begin{proof}
  By induction on $n:\NN$. The base case $n\equiv 0$ is trivial.
  The step from $n$ to $n+1$ is an application of \cref{lem:ap-Sigma},
  using that $\Sn {n+1} \jdeq \susp {\Sn n}$ and $\loopspace {n+1}(f) \jdeq
  \loopspace n (\loopspace\null (f))$.
\end{proof}

The above lemmas allow us to formulate a connection between 
the maps $\susp$ in \cref{thm:susp-is-monoid-iso} and the unit 
$\eta_{\Sn n}$ from \cref{prop:susp-loop-adjunction}. 
The following diagram commutes by \cref{lem:adj-prop} (top triangle)
and \cref{lem:iterated-ap-Sigma} (bottom quadrangle).
This proves \cref{lem:sigma-susp}.
\begin{equation}
    \begin{tikzpicture}[x=6cm,y=-2cm,baseline=(current bounding box.center)]
    \node (A) at (0,0) {$(\Sn n \to_* \Sn n)$};
    \node (B) at (0,1) {$\Omega^n(\Sn n)$};
    \node (C) at (1,1) {$\Omega^{n+1}(\Sn {n+1})$};
    \node (M) at (1,0.5) {$(\Sn n \to_* \Omega\susp\Sn n)$};
    \node (D) at (1,0) {$(\Sn {n+1} \to_* \Sn {n+1})$};

    \draw[arrow] (A) to node [left] {$\phi_{\Sn n}^n$} (B);
    \draw[arrow] (B) to node [above] {$\Omega^n(\eta_{\Sn n})$} (C);
    \draw[arrow, shorten <= -0.5ex, shorten >= -0.5ex] (D) to node [right] {$\Phi_{\Sn n,\Sn {n+1}}$} (M);
    \draw[arrow, shorten <= -0.5ex, shorten >= -0.5ex] (M) to node [right] {$\phi_{\Omega\Sn {n+1}}^n$} (C);
    \draw[arrow] (A) to node [above] {$\susp$} (D);
    \draw[arrow] (A) to node [below] {$\eta_{\Sn n}\circ\blank$} (M);
    \end{tikzpicture} \nonumber
\end{equation}

\begin{restatable}{cor}{lemmasigmasusp}\label{lem:sigma-susp}
$\phi_{\Sn {n+1}}^{n+1}\circ\susp \jdeq
\phi_{\Omega\Sn {n+1}}^n\circ \Phi_{\Sn n,\Sn {n+1}}\circ\susp =
\Omega^n(\eta_{\Sn n})\circ\phi_{\Sn n}^n$ 
\end{restatable} 

Since the $\phi$'s in \cref{lem:sigma-susp} are equivalences we can
transport knowledge about $\Omega^n(\eta_{\Sn n})$ to $\susp$,
using the following result:

\begin{thm}[{Freudenthal suspension theorem \cite[Thm.~8.6.4]{HoTT}}] \label{thm:freudenthal}
    If $X$ is $n$-connected and pointed, with $n \geq 0$, 
    then the map $\eta_X : X \ptdto \loopspace \null \susp X$ 
    is $2n$-connected. (Note $\eta_X \jdeq \sigma_X$ in \cite{HoTT}.)
\end{thm}

Taking $X\jdeq\Sn n$ and using that $\Sn n$ is $(n{-}1)$-connected 
\cite[Cor.~8.2.2]{HoTT}, we get in particular the following instantiation:

\begin{cor} \label{cor:sigma-truncated}
The map $\eta_{\Sn n} : \Sn n \to \Omega(\Sn {n+1})$
is a $2(n-1)$-connected for all $n \geq 1$,
\end{cor}
To make use of this, we show how connectedness of functions interacts with loop spaces (proof in appendix):

\begin{restatable}{lem}{lemmaconnapf} \label{lem:conn-ap}
    Let $A$ and $B$ be types and $f : A \to B$ be a $k$-connected 
    function ($k \geq -1$).
    For all $a_1, a_2 : A$, the function 
    $\ap f : a_1 = a_2 \to f(a_1) = f(a_2)$ is $(k{-}1)$-connected.
\end{restatable}
Iterating this $n$ times combined with~\cref{cor:sigma-truncated} we get:
\begin{cor}
   \label{cor:n-2-connected}
   The map
   \begin{equation}\nonumber
     \loopspace n (\eta_{\Sn n}) : 
     \Omega^n(\Sn n) \to \Omega^{n+1}(\Sn {n+1})
   \end{equation}
   is $(n{-}2)$-connected for all $n\geq 1$.
\end{cor}

We are now ready to make good on our promise made above.
\begin{proof}[Proof of \cref{thm:susp-is-monoid-iso}.]
  Let $n\geq 1$.  We have to prove that
  $\susp : (\Sn n \to_* \Sn n) \to (\Sn {n+1} \to_* \Sn {n+1})$ is a
  $0$-equivalence. It is a wild monoid morphism by functoriality of $\susp$, see
  \cref{ex:loop-sus-wild-functors}(ii), since identity and composition
  of the monoid structures are just given by the identity function and
  function composition.

  We now show that $\susp$ above is a $0$-equivalence.
  By \cref{cor:n-2-connected} we have that $\loopspace n (\eta_{\Sn n})$
  is $(n-2)$-connected. For $n\ge 2$, it is then directly a $0$-equivalence.
  For $n=1$, only have that $\eta_\Sc$ is $0$-connected,
  so induces a surjection on fundamental groups.
  But $\eta_\Sc$ has a retraction $\tau$ by~\cref{lem:tau-retraction-unit-circle},
  so it also induces an injection, hence a bijection, on fundamental groups,
  so $\loopspace\null(\eta_\Sc)$ is also a $0$-equivalence.

  It follows by \cref{lem:sigma-susp}
  that also $\susp$ is a $0$-equivalence, since the $\phi$'s there
  are equivalences.
\end{proof}

\subsection{Connected components of
  \texorpdfstring{$\Sn n = \Sn n$}{Sn = Sn}}
\label{sec:conncomps-Sn=Sn}

\cref{thm:susp-is-monoid-iso} implies that, for all $n \geq 1$, 
the map $\higherTrunc 0 {\susp} :  \higherTrunc 0 {\Sn n \to_* \Sn n} \to \higherTrunc 0 {\Sn {n+1} \to_* \Sn {n+1}}$
is an isomorphism of monoids.
Two more, smaller steps are needed to be able to determine the number of components of $\Sn n \simeq \Sn n$.
One is to remove the base points of this monoid morphism. The other is to consider equivalences
rather than just maps  $\Sn n \to \Sn n$.

\Cref{lem:forget-points-conn} allows us to remove the point in the (co)domain of $\susp$:
\begin{restatable}{lem}{lemmahighSnconn} \label{lem:higher-spheres-connected}
	For all natural numbers $n \geq 1$, the wild monoid morphism
$\susp : (\Sn n \to \Sn n) \to (\Sn {n+1} \to \Sn {n+1})$
  is a $0$-equivalence.
\end{restatable}
\begin{proof}
We have $\pr_1\circ\susp_* = \susp\circ\pr_1$, with $\susp_*$ the suspension morphism
for pointed maps. Since every $\Sn n$ (for $n \geq 1$) is $0$-connected, we get
that $\pr_1$ is $0$-connected by \cref{lem:forget-points-conn}.
It follows that $\susp$ is a $0$-equivalence.
\end{proof}

\begin{cor} \label{cor:higher-spheres-2-comp}
	For all natural numbers $n \geq 1$, the map
	\begin{equation}
	\higherTrunc 0 {\susp} : \higherTrunc 0 {\Sn n \to \Sn n} \to \higherTrunc 0 {\Sn {n+1} \to \Sn {n+1}}
	\end{equation}
	is an isomorphism of monoids.
\end{cor}


The second step is to consider equivalences
rather than just maps  $\Sn n \to \Sn n$.
The following result implies that $\setTrunc{\Sn n = \Sn n}$ is the group of 
invertible elements of the monoid $\setTrunc{\Sn n \to \Sn n}$.
The proof is in the appendix.

\begin{restatable}{lem}{lemmainvertables}
    Let $A$ be a type. Then $\setTrunc{A \weq A}$ is equivalent to the set of 
invertible elements in the monoid $\setTrunc{A \to A}$.
    \label{lemma:invertible-truncated-Sn-pointed-weq}
\end{restatable}

We now see that there are two connected components of symmetries of spheres:
\begin{thm} \label{thm:higher-spheres-2}
    For any $n \geq 1$, we have an equivalence of types
    \begin{equation}   \nonumber
    \higherTrunc 0 {\Sn n= \Sn n} \weq \bn 2
    \end{equation}
\end{thm}
\begin{proof}
  For $n = 1$ and $n=2$, we have established this result in the previous
  sections (see \cref{thm:symmetries-of-S1} and
  \cref{prop:symm-S2-connected-components}).  For higher $n$, it follows by
  induction on $n$ with the help of \cref{cor:higher-spheres-2-comp} that the
  monoids $\setTrunc{\Sn n \to \Sn n}$ have exactly two invertible elements.
  Then, \cref{lemma:invertible-truncated-Sn-pointed-weq} allows us to conclude.
\end{proof}

\subsection{Concrete symmetries of \texorpdfstring{$\Sn n$}{Sn}}
\label{sec:-id<>id-Sn=Sn}

As in the cases of $\Sc$ and $\Sp$, we want to construct one concrete
element for each of the two connected components of symmetries of $\Sn n$
that were established in \cref{thm:higher-spheres-2}.
As before, we take $\id_{\Sn n}$ in one component.
For any type $A$, define $-\id_{\susp(A)}$ by
$\north_{\susp(A)} \mapsto \south_{\susp(A)}$,
$\south_{\susp(A)} \mapsto \north_{\susp(A)}$,
and $\mrd_{\susp(A)}(a) \mapsto \inv{\mrd_{\susp(A)}(a)}$.
Clearly, being self-inverse, $-\id_{\susp(A)}$ is a symmetry of $\susp(A)$.
What is less obvious is that $-\id_{\Sn n}$ is indeed in the
other component of $\Sn n =\Sn n$. Preparing for a proof of
$\id_{\Sn n} \neq -\id_{\Sn n}$ by induction on $n\geq 1$, 
we show that suspension and ``negation'' commute.

\begin{restatable}{lem}{lemmasuspnegcomm} \label{lem:susp-neg-commute}
For any type $A$, the two functions $\susp(-\id_{\susp(A)})$ and 
$-\id_{\susp(\susp(A))}$ of type $\susp(\susp(A)) \to \susp(\susp(A))$ 
are equal.
\end{restatable}
\begin{proof}
  \def\cg{\color{darkgreen}}
  \def\cb{\color{darkblue}}
  For brevity we abbreviate $N \defequi \north_{\susp \susp A}$
  and $S \defequi \south_{\susp \susp A}$.
  We construct $h(x)$ of type 
  $T(x) \defequi ( \susp(-\id_{\susp A})(x) = -\id_{\susp{\susp A}}(x) )$
  by induction on $x:\susp\susp A$ setting 
  $h(N) \defequi \mrd(\north_{\susp A}): (N=S)$ and
  $h(S) \defequi \inv{(\mrd(\south_{\susp A}))}: (S=N)$.
  To complete the definition of $h$ (and the proof of the lemma)
  we need to define, for all $y : \susp A$,
  a higher path $\ap h(\mrd(y))$ whose type is
  $\pathover {h(N)} T {\mrd(y)} {h(S)}$. By \cite[Lem.~2.11.3]{HoTT},
  and using abbreviations 
  $m_N \defequi \mrd(\north_{\susp A})$ and
  $m_S \defequi \mrd(\south_{\susp A})$,
  the latter type is equivalent to $$U(y) \defequi
  (\inv{m_S} \cdot \mrd(-\id_{\susp A}(y)) = \inv{\mrd(y)} \cdot m_N).$$
  
  We construct $g(y)$ of type $U(y)$ by induction on $y: \susp A$.
  The type $U(\north_{\susp A})$ is, after simplification (normalising), 
  $\inv{m_S} \cdot m_S = \inv{m_N} \cdot m_N$.
  For $g(\north_{\susp A})$ we take the path $c(m_S,m_N)$, 
  where $c$ is defined by double path induction, 
  setting $c(\refl{},\refl{}) \jdeq \refl{\refl{}}$.
  Similarly, $U(\south_{\susp A}) \defequi (\inv{m_S} \cdot m_N =
  \inv{m_S} \cdot m_N)$, and we take $g(\south_{\susp A}) \defequi
  \refl{\inv{m_S} \cdot {m_N}}$.  
  To complete the definition of $g$ we need to define, for all $z : A$,
  a higher path $\ap g(\mrd(z))$ whose type is
  $\pathover {c(m_S,m_N)} U {\mrd(z)} {\refl{\inv{m_S} \cdot {m_N}}}$.
  In general, transport of an arbitrary $c:U(y)$ along 
  $p:y=y'$ in $\susp A$ yields a $c': U(y')$ given by
        \begin{displaymath}
          \ap {\blank \cdot m_N} (\ap {\inv{} \circ \mrd} (p))
          \cdot c \cdot
          \inv{\left( \ap{\inv{m_S} \cdot \blank} 
          (\ap {\mrd \circ -\id_{\susp A}} (p)) \right)}
        \end{displaymath}
  This follows again from \cite[Lem.~2.11.3]{HoTT},
  or by path induction on $p$.
  
  We now instantiate this transport with $p \jdeq \mrd(z)$ for $z:A$,
  and abbreviate $\beta \defequi \ap\mrd(\mrd(z)) : m_N=m_S$.
  By unfolding the definition of $-\id_{\susp A}$ and path algebra
  we get $\ap{\mrd}(\ap{-\id_{\susp A}}(\mrd(z))) = \inv{\beta}$. 
  Here $\ap\mrd \jdeq \ap{\mrd_{\susp\susp A}}$ and
  $\mrd(z)\jdeq\mrd_{\susp A}(z)$. Further calculations show that
  to define $\ap g(\mrd(z))$ it suffices to find an element of
	\begin{equation}\label{eq:first-square-composed}
          \ap{(\blank\cdot {m_N})\circ{}^{-1}}(\beta) \cdot c(m_S,m_N)
               \cdot \ap{\inv{m_S}\cdot\blank}(\beta)
                            =
          \refl{\inv{m_S} \cdot {m_N}} \nonumber
	\end{equation}
	of type $\inv{m_S} \cdot m_N = \inv{m_S} \cdot m_N$.
  In other words, we should fill the following diagram of 2-paths:
  \begin{equation}\nonumber
  \begin{tikzcd}[column sep=4em,labels={font=\normalsize}]
    m_N^{-1}\cdot {m_N}
      \arrow[Rightarrow, swap]{dd}{\ap{(\blank\cdot {m_N})\circ\inv{}}(\beta)}
      \arrow[Leftarrow]{rr}{c(m_S,m_N)} && {\inv{m_S}}\cdot m_S
      \arrow[Leftarrow]{dd}{\ap{\inv{m_S}\cdot\blank}(\beta)}
 \\\\
    \inv{m_S}\cdot {m_N}
       \arrow[Leftarrow,swap]{rr}{\refl{\inv{m_S} \cdot {m_N}}} &&
    {\inv{m_S}}\cdot m_N
  \end{tikzcd}
  \end{equation}
  The easiest way to fill the above diagram is to abstract
  completely from suspension types to a type $T$
  with points $N,S : T$, paths $m_N,m_S : N=S$,
  and a 2-path $\beta: m_N=m_S$.
  One starts by doing path induction on $m_N$,
  reducing the task to the case $S\jdeq N$ and $m_N\jdeq \refl N$,
  and arbitrary $m_S:N=N$ and $\beta: m_N = m_S$.
  One can then do path induction on $\beta$,
  reducing the task to the case $m_S\jdeq m_N$
  and $\beta \jdeq \refl{m_N}$.
  But now, as $m_N \jdeq \refl N$, all paths appearing in the diagram,
  including $c(m_S,m_N)$, are reflexivity paths.
  Hence we can conclude by simple path algebra.
\end{proof}
A formal proof of this lemma is available in cubical Agda~\cite{doublesusp}.

\begin{cor}\label{cor:id-not-minus-id}
  For any $n \geq 1$, we have $\id_{\Sn n} \neq -\id_{\Sn n}$, and any
  symmetry of $\Sn n$ is merely equal to either $\id_{\Sn n}$ or
  $-\id_{\Sn n}$.
\end{cor}
\begin{proof}
  We know the statement for $n \equiv 1$ (\cref{eq:id-neq-minusid-Sc})
  and $n \equiv 2$ (\cref{lemma:S2-id-neq-minusid}) from the previous
  sections. The rest is done by induction
  on $n$, so we assume $\id_{\Sn n} \neq -\id_{\Sn n}$.  By
  \cref{cor:higher-spheres-2-comp}, we have
  $\susp(\id_{\Sn n}) \neq \susp(-\id_{\Sn n})$.  As
  $\id_{\Sn {n+1}} = \susp(\id_{\Sn n})$ trivially holds and we
  further have $-\id_{\Sn {n+1}} = \susp(-\id_{\Sn n})$ by
  \cref{lem:susp-neg-commute}, the claimed inequality follows.
  Therefore, the two symmetries lie in different components, and any
  symmetry lies in one of the two components given by
  \cref{thm:higher-spheres-2}.
\end{proof}

\section{Summary and comparison}
\label{sec:summary}

In \cref{fig:comparisons} we depict some relationships between the types studied so far.
In the back, we see the types of pointed maps $\Sn n \ptdto \Sn n$ on the top,
the types of maps $\Sn n\to\Sn n$ in the middle,
and the types of identifications $\Sn n = \Sn n$ on the bottom,
each related by suspension as we go left-to-right.
In the front, we see the set truncations thereof,
with the set truncation maps going back-to-front.
Additionally, we see on the front left concrete monoids that are equivalent
to the types involving the $0$-sphere.
(The map $-0 : \Sn0\to\Sn0$ is the constant map at the non-base point;
this becomes identified with $0$ after one suspension.)

On the right, we see the sequential colimits.
(Sequential colimits can be defined using pushouts and coproducts over $\NN$
or as a HIT as in~\cite[Sec.~3]{seqcolim}.)
The dotted arrows are lifts of the top back squares
since we form the suspension of a pointed map by first forgetting the pointedness.
The top two sequences in the back thus sit cofinally inside the zigzagging sequence,
and hence they have the same colimit, which we identify with the elements
of the sphere spectrum, $\loopspace\infty\mathbb{S}$.
(The sphere spectrum $\mathbb{S}$ is the spectrification of the prespectrum of spheres,
i.e., $\Sigma^\infty\Sn0$, where $\Sigma^\infty$ maps a pointed type
to the corresponding suspension spectrum,
and $\loopspace\infty$ is the right adjoint thereof, which maps a spectrum
to its underlying infinite loop type.
See~\cite[Sec.~5.3]{vandoorn:thesis} for more on spectra in HoTT.)
The sequential colimit of the self-identification groups
$\mathrm{G}(n+1) \jdeq (\Sn n=\Sn n)$
is the group of units of the sphere spectrum,
$\mathrm{G} \jdeq \mathrm{GL}_1(\mathbb{S}) = (\mathbb{S} = \mathbb{S})$.

The diagram commutes, with the exception of the dashed arrow.
This takes a pointed map $f : \Sp\ptdto\Sp$
to the composite $\tau \circ \loopspace\null(f) \circ \eta_{\Sc} : \Sc\ptdto\Sc$,
where $\tau : \loopspace\null\Sp \ptdto \Sc$ comes from the H-space structure
on the circle, cf.~\cref{eq:tau}.
This is a retraction of the suspension operation, because
if $f \jdeq \susp(g)$, then
\[
  \tau \circ \loopspace\null \susp(g) \circ \eta_{\Sc}
  = \tau \circ \eta_{\Sc} \circ g = g
\]
by naturality of $\eta$ and~\cref{lem:tau-retraction-unit-circle}.
Relative to the equivalences $(\Sc\ptdto\Sc)\weq\loopspace\null\Sc$
and $(\Sp\ptdto\Sp)\weq\loopspace2\Sp$, the dashed map
can be identified with
$\loopspace\null(\tau) : \loopspace2\Sp \ptdto \loopspace\null\Sc$,
which is a $0$-equivalence.

The homotopy groups of $\loopspace\infty\mathbb{S}$ are of course
the stable homotopy groups of the spheres, $\hgr k^{\mathrm s}$.
Since the type of units $\mathrm{G}$ embeds into $\loopspace\infty\mathbb{S}$,
it has the same homotopy groups, except for the connected components:
\[
  \hgr k(\mathrm{G}) =
  \begin{cases}
    \ZZ/2\ZZ, &\text{for $k=0$,} \\
    \hgr k^{\mathrm s}, &\text{otherwise.}
  \end{cases}
\]

Though not shown in \cref{fig:comparisons}, the types of pointed identifications
$\Sn n=_*\Sn n$ are equivalent to pullbacks of the vertical cospans in the back.
We have embeddings
$(\Sn n=_*\Sn n) \hookrightarrow \loopspace n\Sn n$ onto the subtypes
$\Omega^n_{\pm1}\Sn n$ corresponding to the generators
$\{\pm1\}$ in $\hgr n\Sn n\simeq \ZZ$. Hence the homotopy groups
of $\Sn n=_*\Sn n$ are the usual (unstable) homotopy groups of spheres,
except at degree $0$:
\[
  \hgr k(\Sn n=_* \Sn n,\pm\id) =
  \begin{cases}
    \ZZ/2\ZZ, &\text{for $k=0$,} \\
    \hgr{n+k}\Sn n, &\text{otherwise.}
  \end{cases}
\]

Our construction of degree functions establishing the above picture
follows in most respects the approach outlined in \cite[Sec.~5]{Buchholtz2018CellularCI} and formalized in~\cite{hott-agda}.
The difference is that we focus on the (wild) monoid structures given by composition
instead of the group structures given either by the cogroup structures on the spheres
or by transport from the wild groups $\loopspace n\Sn n$.
We thus give a direct proof of~\cref{cor:id-not-minus-id} that is interesting
its own right,
even though one could also conclude that $d(-\id_{\Sn n})=-1$ from
the facts that $-\id_{\Sn n}$ is the additive inverse of $\id_{\Sn n}$
and that the degree is a $0$-equivalence sending $\id_{\Sn n}$ to $1$.

\section{Interlude on Whitehead products}
\label{sec:whitehead-interlude}

Before focusing on the components of $\Sn n = \Sn n$,
we need a few general results on Whitehead products. 
Recall from \cite[Chapter 6.8]{HoTT} the join $*$ and the wedge $\vee$,
higher inductive operations on types that can be constructed
using pushouts. The proof of the following lemma is in the appendix.

\begin{restatable}{lem}{lemmaUMPjoinvsloop}
  Let $A$, $B$, and $X$ be pointed types.
  We have the following equivalence:
   $(A \ptdto (B \ptdto \loopspace{}X)) \weq (A*B \ptdto X)$.
  \label{lemma:UMP-join-vs-loopspace}
\end{restatable}

For the rest of this section we fix two pointed types $A$ and $B$, and we
denote $a_0$ and $b_0$ the base points of $A$ and $B$.
We repeat here the definition of the generalized Whitehead product
from \cite[Sec.~3.3]{brunerie:thesis}.
\begin{defi}
  Define the map $W = W_{A,B} : A * B \to \susp A \vee \susp B$
  by $W(\inl(a)) \defeq \inr(\north_{\susp B})$, 
  $W(\inr(b)) \defeq \inl(\north_{\susp A})$,
  and
  \[
    \ap{W}(\glue(a,b)) \defis \ap{\inl}(\eta_A(a))
    \cdot (\glue(\ast))^{-1} \cdot \ap{\inr}(\eta_B(b))
  \]
  ($\eta$ from \ref{prop:susp-loop-adjunction}, 
  $\glue(\ast): \inl(\north_{\susp A}) = \inr(\north_{\susp B})$).
   The map $W_{A,B}$ is pointed by the path 
   $\inv{(\glue(\ast))} : W(\inl(a_0)) \jdeq \inr(\north) = \inl(\north)$.
\end{defi}   
  The map $W$ makes a pushout square with the wedge inclusion $i$:
  \begin{equation}\label{eq:def-join-wedge-pushout}
    \begin{tikzcd}
      A * B \ar[r,"W"]\ar[d]\ar[pushout] & \susp A \vee \susp B\ar[d,"i"] \\
      1 \ar[r] & \susp A \times \susp B
    \end{tikzcd}
  \end{equation}
The fact that \eqref{eq:def-join-wedge-pushout} is a pushout square is
deduced using the $3{\times}3$-lemma in \cite[Prop.~3.3.2]{brunerie:thesis},
but plays no role in our arguments below.

We now also fix another pointed type $(X,x_0)$.
\begin{defi}
  \label{defn:gen-Whitehead-product}
  The generalized Whitehead product of $\alpha : \susp A \ptdto X$
  and $\beta : \susp B \ptdto X$
  is the composition
  \[
    [\alpha, \beta] \defeq (\alpha \vee \beta) \circ W_{A,B}
    : A * B \ptdto X.
  \]
  Here $\alpha \vee \beta \defeq \ind(\alpha,\beta,\inv \beta_0 \alpha_0)$
  by $\vee$-induction, with $\alpha_0,\beta_0$ the pointing paths, so
  $\inv \beta_0 \alpha_0 : \alpha(\north_{\susp A}) = \beta(\north_{\susp B})$.
\end{defi}
\begin{rem}\label{rem:whitehead-products-pis}
  If $A\jdeq \Sn p$ and $B\jdeq\Sn q$,
  then $A * B \weq \Sn{p+q+1}$~\cite[Prop.~1.8.8]{brunerie:thesis},
  so one can obtain a map (using the same denotation)
  \begin{equation*}
    [\blank,\blank] : \hgr {p+1}(X) \times \hgr {q+1}(X) \to \hgr{p+q+1}(X),
  \end{equation*}
  which is the usual Whitehead product on homotopy groups. Explicitly, given
  $a:\hgr {p+1} (X)$ and $b:\hgr {q+1} (X)$, one wants to define the element
  $[a,b]$ of the set $\hgr {p+q+1} (X)$. As we are targeting a set, one can as
  well assume $a \jdeq \settrunc{\alpha}$ and $b \jdeq \settrunc{\beta}$ for
  $\alpha:\loopspace{p+1}(X)$ and $\beta:\loopspace{q+1}(X)$. Using the
  equivalences
  \begin{equation*}
    \phi_X^{p+1} :  \left(\susp\Sn {p} \ptdto X \right) \weq \loopspace{p+1}(X), \quad
    \phi_X^{q+1} :  \left(\susp\Sn {q} \ptdto X \right) \weq \loopspace{q+1}(X),
  \end{equation*}
  one gets pointed maps
  $\inv{(\phi_X^{p+1})}(\alpha): \susp \Sn {p} \ptdto X$ and
  $\inv{(\phi_X^{q+1})}(\beta): \susp \Sn {q} \ptdto X$. The element
  $[a,b] : \hgr{p+q+1}(X)$ is then defined as
  $\settrunc{[\inv{(\phi_X^{p+1})}(\alpha),\inv{(\phi_X^{q+1})}(\beta)]}$
  (where the bracket follows \cref{defn:gen-Whitehead-product}).
\end{rem}
Fix now a pointed map $\beta : \susp B \ptdto X$ with pointing path
$\beta_0:\beta(\north) = x_0$, and consider the fiber sequence
of evaluation at $\beta$:
\[
  \conncomp{(\susp B \ptdto X)}{\beta} \stackrel \iota \longrightarrow_\ast
  \conncomp{(\susp B \to X)}{\beta} \stackrel {\ev_\beta} \longrightarrow_\ast
  X.
\]
(Taking connected components is not necessary, but
allows us to emphasize the base points of the various function types.)
This fiber sequence induces a long exact sequence:
\begin{equation*}
\begin{tikzpicture}[x=2.5cm,y=-1cm,baseline=(current bounding box.center)]
	\node (A) at (-.3,0) {$\cdots$};
	\node (B) at (.7,0) {$\hgr{n+1}(\susp B \to X, \beta)$};
	\node (C) at (2,0) {$\hgr{n+1}(X)$};
	\node (D) at (0,1) {$\hgr n(\susp B \ptdto X,\beta)$};
	\node (E) at (1.3,1) {$\hgr n(\susp B \to X,\beta)$};
	\node (F) at (2.3,1) {$\cdots$};

	\draw[arrow] (A) to node [above] {} (B);
	\draw[arrow] (B) to node [above] {$\hgr {n+1} (\ev_\beta)$} (C);
	\draw[arrow,rounded corners] 
		(C) -| 
		node[auto,text=black,pos=.9] {$\partial^n_\beta$}
		($(C.east)+(.25,.5)$) |-
	 	($(C)!.5!(D)$) -|
    	($(D.west)+(-.25,0)$) |- (D);
	\draw[arrow] (D) to node [below] {$\hgr n (\iota)$} (E);
	\draw[arrow] (E) to node [above] {} (F);

\end{tikzpicture}
\end{equation*}
The construction presented in \cite[Ch.~8.4]{HoTT} of this exact
sequence is as follows. Consider the map
$\kappa_\beta : \loopspace \null {X} \to \conncomp {(\susp B \ptdto
  X)} \beta$ that associates to a loop $\alpha$ the function
$\beta$ pointed by the path $\alpha \cdot \beta_0$. Then the long fiber
sequence is shown equivalent to the one in~\cref{fig:LES}.
\begin{figure*}
\caption{Long fiber sequence of evaluation.}\label{fig:LES}
\begin{minipage}{\textwidth}
\begin{displaymath}
  \begin{tikzcd}[column sep=huge]
    \cdots \ar[r, ptdto] &
    \loopspace {n+1} {\left(\conncomp{(\susp B \to X)}{\beta}\right)}
    \ar[r, "(-1)^{n+1}\loopspace {n+1} (\ev_\beta)"{yshift=7pt}, ptdto] \ar[d, phantom, ""{coordinate, name=Z}] &
    \loopspace {n+1} X \ar[dll,"(-1)^n\loopspace n {(\kappa_\beta})"{xshift=-10ex}, rounded corners, shorten >=.25em,
    to path={ -- ([xshift=8.5ex]\tikztostart.center)
      |- (Z)
      -| ([xshift=-16ex]\tikztotarget.center) [near start,swap]\tikztonodes -- node[at end,yshift=-3pt]{\scriptsize$\ast$} (\tikztotarget)}] \\
    \loopspace {n} {\left(\conncomp{(\susp B \to X)}{\beta}\right)}
    \ar[r, "(-1)^n\loopspace n (\iota)", ptdto] &
    \loopspace n {\left( \conncomp{(\susp B \to X)}{\beta} \right)} \ar[r,ptdto] &
    \cdots
  \end{tikzcd}
\end{displaymath}
\end{minipage}
\end{figure*}
It is then shown that the set-truncation of this sequence is a long
exact sequence of sets, and the very last paragraph of the proof of
\cite[Thm.~8.4.6]{HoTT} proceeds to replace the truncations of the form
$\setTrunc {-\loopspace n (h)}$ (which are group antimorphisms) by
$\hgr n (h)$ (which are actual group morphisms) for $n\ge1$. In
particular, the boundary map $\partial^n_\beta$ in that long exact
sequence can be taken to be $\hgr n {(\kappa_\beta)}$.

However, we wish to express $\partial^n_\beta$ at $n\ge0$ in terms of
the Whitehead product.
In order to do so, we use the equivalences $\phi_A^n$ from
\cref{def:sphere-to-eq-loops} and reason about the map
$\delta^n_\beta: (\susp \Sn n \ptdto X) \to (\Sn n \ptdto
\conncomp{(\susp B \ptdto X)}\beta)$ defined by
\begin{equation*}
  \delta^n_\beta 
  \defequi  \inv{\left( \phi_{\conncomp{(\susp B \ptdto X)}\beta}^n \right)} \circ \loopspace n (\kappa_\beta) \circ \phi_X^{n+1}
\end{equation*}
so that in particular
$\partial^n_\beta = \setTrunc{{\phi_{\conncomp{(\susp B \ptdto
        X)}\beta}^n} \circ \delta^n_\beta \circ
  \inv{(\phi_X^{n+1})}}$. Notice that this function $\delta^n_\beta$
has a simple expression through the use of \cref{lem:iterated-ap-Sigma}:
\begin{align*}
  \delta^n_\beta
  &\jdeq \inv{\left( \phi_{\conncomp{(\susp B \ptdto X)}\beta}^n \right)} \circ \loopspace n (\kappa_\beta) \circ \phi_X^{n+1} \\
  & = \inv{\left( \phi_{\conncomp{(\susp B \ptdto X)}\beta}^n \right)} \circ \loopspace n (\kappa_\beta) \circ \phi_X^{n} \circ \Phi_{\Sn n, X} \\
  & = \kappa_\beta \circ (\Phi_{\Sn n, X} (\blank) )
\end{align*}
In other words, for every $\alpha : \susp \Sn n \ptdto X$ and
$x: \Sn n$, the pointed map $\delta^n_\beta (\alpha)(x)$ is just
$\beta$ as an unpointed function, but is pointed by the path:
\begin{equation*}
  (\loopspace\null(\alpha)\circ\eta_{\Sn n})(x) \cdot \beta_0
  : \beta(\north) = x_0
\end{equation*}

To express $\delta^n_\beta$ in terms of
generalized Whitehead products, we are going to construct a commuting square of
the following form:
\begin{equation*}
  \begin{tikzcd}
    (\susp{\Sn n} \ptdto X)
    \ar[r,"\delta^n_\beta"] \ar[d,"\rho^n_\beta"'] &
    (\Sn n \ptdto \conncomp{(\susp B \ptdto X)}{\beta})
    \ar[d,"\xi^n_\beta\circ\blank","\sim" rotninety] \\
    (\Sn n * B \ptdto X) &
    (\Sn n \ptdto \conncomp{(\susp B \ptdto X)}{\wunit}) \ar[l,"\varphi^n"] \ar[l,"\sim"']
  \end{tikzcd}
\end{equation*}
where the maps $\xi^n_\beta,~\rho^n_\beta$ and $\varphi^n$ are to be
defined, and the element of $\susp B \ptdto X$ denoted $\wunit$ is the
constant map at $x_0$.

\newcommand*\miniast[2]{\nu_{#1,#2}}
Since nothing hinges on having a sphere $\Sn n$, let us generalize
and construct a commuting diagram for any pointed connected type $A$:
\begin{equation}\label{eq:Whitehead-general}
  \begin{tikzcd}
    (\susp A \ptdto X)
    \ar[r,"\delta_\beta"] \ar[d,"\rho_\beta"'] &
    (A \ptdto \conncomp{(\susp B \ptdto X)}{\beta})
    \ar[d,"\xi_\beta\circ\blank","\sim" rotninety] \\
    (A * B \ptdto X) &
    (A \ptdto \conncomp{(\susp B \ptdto X)}{\wunit}) \ar[l,"\sim"']\ar[l,"\varphi"]
  \end{tikzcd}
\end{equation}
where $\delta_\beta$ is defined on $\alpha$ as follows: for $a:A$,
$\delta_\beta(\alpha)(a) \jdeq \beta$ as an unpointed function,
pointed by the path:
\begin{equation*}
  (\delta_\beta(\alpha)(a))_0 \defequi
  (\loopspace\null(\alpha)\circ\eta_A)(a) \cdot \beta_0
  = \alpha_0 \cdot \ap{\alpha}(\eta_A(a)) \cdot \miniast\alpha\beta
  : \beta(\north) = x_0
\end{equation*}
Here we write
$\miniast\alpha\beta \defequi \alpha_0^{-1} \cdot \beta_0 :
\beta(\north) = \alpha(\north)$ for short, where
$\alpha_0:\alpha(\north) = x_0$ is the path pointing $\alpha$. Notice
that the map $\delta_\beta(\alpha)$ is indeed a pointed map: the
element $\delta_\beta(\alpha)(a_0)$ is the map $\beta$ pointed by the
path
$\alpha_0\cdot \ap\alpha(\eta_A(a_0)) \cdot \miniast \alpha \beta$;
using the fact that $\eta_A(a_0) = \refl\north$, one finds a path
$(\delta_\beta(\alpha)(a_0))_0 = \beta_0$, providing an element of
$\delta_\beta(\alpha)(a_0) = \beta$ as pointed functions.

Next define $\rho_\beta \defequi [\blank,\beta]$ as the Whitehead product, or
explicitly:
\begin{align*}
  \rho_\beta(\alpha)\,(\inl(a)) &\jdeq \beta(\north) \\
  \rho_\beta(\alpha)\,(\inr(b)) &\jdeq \alpha(\north) \\
  \ap{\rho_\beta(\alpha)} (\glue(a,b)) &= \ap{\alpha}(\eta_A(a)) \cdot
                                         \miniast\alpha\beta \cdot
                                         \ap{\beta}(\eta_B(b))
\end{align*}
The map $\rho_\beta(\alpha)$ is pointed by the path $\beta_0 :
\rho_\beta(\alpha) (\inl(a)) \jdeq \beta(\north) = x_0$.

Let us now describe $\xi_\beta$. Since $\susp B \ptdto X$ is equivalent to $B \ptdto \loopspace\null X$, and because $\loopspace\null X$ is a wild group, so is $\susp B \ptdto X$. Its unit is the map $\wunit\jdeq (\blank\mapsto x_0)$ already described. The multiplication of two elements $\gamma$ and $\gamma'$ is defined by induction:
\begin{equation*}
  \begin{aligned}
    (\gamma' \wadd \gamma) (\north) &\defequi \gamma(\north) \\
    (\gamma' \wadd \gamma) (\south) &\defequi \gamma'(\south) \\
    \ap{\gamma' \wadd \gamma} (\mrd (b) ) &\defis
    \ap{\gamma'}(\mrd(b)) \cdot \miniast{\gamma'}{\gamma} \cdot \ap\gamma(\eta_B(b))
  \end{aligned}
\end{equation*}
The map $\gamma' \wadd \gamma$ is pointed by the path $\gamma_0: \gamma(\north) =
x_0$ pointing $\gamma$ itself.
The inverse $\winv\gamma$ of an element $\gamma$ is given by 
$\gamma \circ (-\id_{\susp B})$
(where $-\id_{\susp B}$ is pointed by $\inv{\mrd(b_0)}$).  Then, there is an
equivalence $(\susp B \ptdto X) \weq (\susp B \ptdto X)$ that maps $\gamma$ to
$\winv\beta \wadd \gamma $, cf.~\cref{prop:equiv-susp-comp}.
This equivalence sends the connected component at
$\beta$ to the connected component at $\wunit$, hence providing the pointed
equivalence $\xi_\beta$. Explicitly:
\begin{equation*}
  \begin{aligned}
    \xi_\beta(\gamma)(\north) &\jdeq \gamma(\north) \\
    \xi_\beta(\gamma)(\south) &\jdeq \beta(\north) \\
    \ap{\xi_\beta(\gamma)}(\mrd b)
    &= \ap{\beta}(\eta_B(b))^{-1} \cdot
    \miniast\beta\gamma \cdot
    \ap{\gamma}(\eta_B(b))
  \end{aligned}
\end{equation*}
Notice that $\xi_\beta(\gamma)$ is pointed by the path $\gamma_0$ that points
$\gamma$.

Let us now define $\varphi$ from diagram \eqref{eq:Whitehead-general}.
Since $A$ is connected, the inclusion of
$A \ptdto \conncomp{(\susp B \ptdto X)}{\wunit}$
in $A \ptdto (\susp B \ptdto X)$ is an equivalence.
Now, use the equivalence between $\susp B \ptdto X$
and $B \to \loopspace\null X$ before
simply applying \cref{lemma:UMP-join-vs-loopspace}.
Unfolding definition, we see that the composition of
equivalences
\begin{equation*}
  \begin{tikzcd}[column sep=large]
    (A \ptdto \conncomp{(\susp B \ptdto X)}{\wunit}) \rar[hookrightarrow]
    & (A \ptdto (\susp B \ptdto X)) \dlar[
      rounded corners,
      to path={
        (\tikztostart.east) -|
        ([xshift=1em, yshift=-2em]\tikztostart.east) -|
        node[near start, fill=white] {\footnotesize$\Phi_{B,X}\circ \blank$}
        ([xshift=-1em]\tikztotarget.west) --
        (\tikztotarget.west)
      }
    ]\\
    (A \ptdto (B \ptdto \loopspace{}X)) \rar["f \mapsto \bar f"']
    & (A*B \ptdto X)
  \end{tikzcd}
\end{equation*}
can be identified with the function $\varphi$ defined by induction as follows:
\begin{equation*}
  \begin{aligned}
  \varphi(h) (\inl(a)) &\defequi x_0 \\
  \varphi(h) (\inr(b)) &\defequi x_0 \\
  \ap{\varphi(h)} (\glue(a,b))
                      &\defis (h(a))_0 \cdot \ap{h(a)}(\eta_B(b))
                        \cdot \inv{(h(a))_0}
  \end{aligned}
\end{equation*}
The map $\varphi(h)$ is pointed by the reflexivity path at $x_0$.


With the preliminaries out of the way, let us show that $\rho_\beta$ can be
identified with $\psi_\beta \defequi \varphi \circ (\xi_\beta \circ \blank)
\circ \delta_\beta$. Let $\alpha : \susp A \ptdto X$.
Unfolding the above definitions, let us first examine
$h \defequi \xi_\beta\circ (\delta_\beta(\alpha)) : 
A \ptdto \conncomp{(\susp B\ptdto X)}{\wunit}$:
\begin{align*}
  h(a)(\north) &\jdeq h(a)(\south) \jdeq \beta(\north) \\
  \ap{h(a)}(\mrd(b))
  &= \ap{\beta}(\eta_B(b))^{-1}
     \cdot \beta_0^{-1}
     \cdot (\delta_\beta(\alpha)(a))_0
     \cdot \ap{\beta}(\eta_B(b)) \\
   &= \ap{\beta}(\eta_B(b))^{-1}
     \cdot \miniast\beta\alpha
     \cdot \ap{\alpha}(\eta_A(a))
     \cdot \miniast\alpha\beta
     \cdot \ap{\beta}(\eta_B(b)) \\
  (h(a))_0
  &= (\delta_\beta(\alpha)(a))_0 \\
  &= \alpha_0 \cdot \ap{\alpha}(\eta_A(a))
           \cdot \miniast\alpha\beta
\end{align*}
Finally, we can insert this into the definition of $\varphi$
to obtain the function $g \defequi \psi_\beta(\alpha) = \varphi(h) : A * B \ptdto X$:
\begin{alignat*}{9}
  g(\inl(a))
  &&&\jdeq &\; &x_0 \\
  g(\inr(b))
  &&&\jdeq &&x_0 \\
  \ap{g}(\glue(a,b))
  &&&=&& (h(a))_0 \cdot \ap{h(a)}(\eta_B(b)) \cdot \inv{(h(a))_0} \\
  &&&=&& (h(a))_0 \cdot \ap{h(a)}(\inv{\mrd(b_0)} \cdot\mrd b)
    \cdot \inv{(h(a))_0} \\
  &&&=&& (h(a))_0 \cdot \inv{\ap{h(a)}(\mrd{b_0})}
    \cdot \ap{h(a)}(\mrd (b)) \\
  &&&&& \phantom{=} \cdot \inv{(h(a))_0} \\
  &&&=&& \bigl(\alpha_0 \cdot \ap{\alpha}(\eta_A(a))
    \cdot \miniast\alpha\beta\bigr)
    \cdot \bigl(\miniast\beta\alpha
    \cdot \ap{\alpha}(\eta_A(a))
    \cdot \miniast\alpha\beta\bigr)^{-1} \\
  &&&&&\phantom{=} \cdot
    \bigl(\ap{\beta}(\eta_B(b))^{-1}
    \cdot \miniast\beta\alpha
    \cdot \ap{\alpha}(\eta_A(a))
    \cdot \miniast\alpha\beta\\
  &&&&&\phantom{=}   \cdot \ap{\beta}(\eta_B(b))\bigr)
    \cdot \bigl(\alpha_0 \cdot \ap{\alpha}(\eta_A(a))\cdot \miniast\alpha\beta\bigr)^{-1} \\
  &&&=&& \beta_0
    \cdot \ap{\beta}(\eta_B(b))^{-1}
    \cdot \miniast\beta\alpha
    \cdot \ap{\alpha}(\eta_A(a))
    \cdot \miniast\alpha\beta\\
  &&&&&\phantom{=}   \cdot \ap{\beta}(\eta_B(b))
    \cdot \miniast\beta\alpha
    \cdot \ap{\alpha}(\eta_A(a))^{-1}
    \cdot \alpha_0^{-1}
\end{alignat*}
The path $g_0$ pointing $g$ is, by definition of $\varphi$,
the reflexivity path at $x_0$.
It only remains to construct a path $H : \rho_\beta(\alpha) = g$ as
mere functions such that the type $H(\inl(a_0)) = \beta_0$ has an
element. We proceed by induction on an element of the join:
\begin{align*}
  H(\inl(a)) &\defequi \alpha_0 \cdot \ap{\alpha}(\eta_A(a))
              \cdot \miniast\alpha\beta : \beta(\north) = x_0 \\
  H(\inr(b)) &\defequi \beta_0 \cdot \ap{\beta}(\eta_B(b))^{-1}
              \cdot \miniast\beta\alpha : \alpha(\north) = x_0
\end{align*}
Finally, we must produce an element of
\begin{equation}\label{eq:big-whitehead-pathover}
  \prod_{a\from A}\prod_{b\from B} \left(
    \pathover {H(\inl(a))} {z\mapsto \rho_\beta(\alpha)(z)=g(z)}
    {\glue(a,b)} {H(\inr(b))}
  \right)
\end{equation}
which corresponds to filling the square in~\cref{fig:big-whitehead-square}.
\begin{figure}
  \[
  \begin{tikzcd}[column sep=large]
    \beta(\north)\ar[r,"{\miniast\alpha\beta}"]
    \ar[ddd,"{\ap{\beta}(\eta_B(b))}"]
    \ar[rrrddd,tikzequal,dashed] &
    \alpha(\north)\ar[r,"{\ap{\alpha}(\eta_A(a))}"]
    \ar[rrdd,tikzequal,dashed] &
    \alpha(\north)\ar[r,"{\alpha_0}"]
    \ar[rd,tikzequal,dashed] &
    x_0\ar[d,"{\alpha_0^{-1}}"] \\
    & & &
    \alpha(\north)\ar[d,"{\ap{\alpha}(\eta_A(a))^{-1}}"] \\
    & & &
    \alpha(\north)\ar[d,"{\miniast\beta\alpha}"] \\
    \beta(\north)\ar[ddd,"{\miniast\alpha\beta}"]
    \ar[rrrd,tikzequal,dashed] & & &
    \beta(\north)\ar[d,"{\ap{\beta}(\eta_B(b))}"] \\
    & & &
    \beta(\north)\ar[d,"{\miniast\alpha\beta}"] \\
    & & &
    \alpha(\north)\ar[d,"{\ap{\alpha}(\eta_A(a))}"] \\
    \alpha(\north)\ar[ddd,"{\ap{\alpha}(\eta_A(a))}"]
    \ar[rrru,tikzequal,dashed] & & &
    \alpha(\north)\ar[d,"{\miniast\beta\alpha}"] \\
    & & &
    \beta(\north)\ar[d,"{\ap{\beta}(\eta_B(b))^{-1}}"] \\
    & & &
    \beta(\north)\ar[d,"{\beta_0}"] \\
    \alpha(\north)\ar[r,"{\miniast\beta\alpha}"']
    \ar[rrruuu,tikzequal,dashed] &
    \beta(\north)\ar[r,"{\ap{\beta}(\eta_B(b))^{-1}}"']
    \ar[rruu,tikzequal,dashed] &
    \beta(\north)\ar[r,"{\beta_0}"']
    \ar[ru,tikzequal,dashed] & x_0
%
%
%
%
%
%
%
%
%
%
%
%
%
%
%
%
%
%
%
  \end{tikzcd}
  \]
  \caption{The square corresponding to \eqref{eq:big-whitehead-pathover}.}
  \label{fig:big-whitehead-square}
\end{figure}
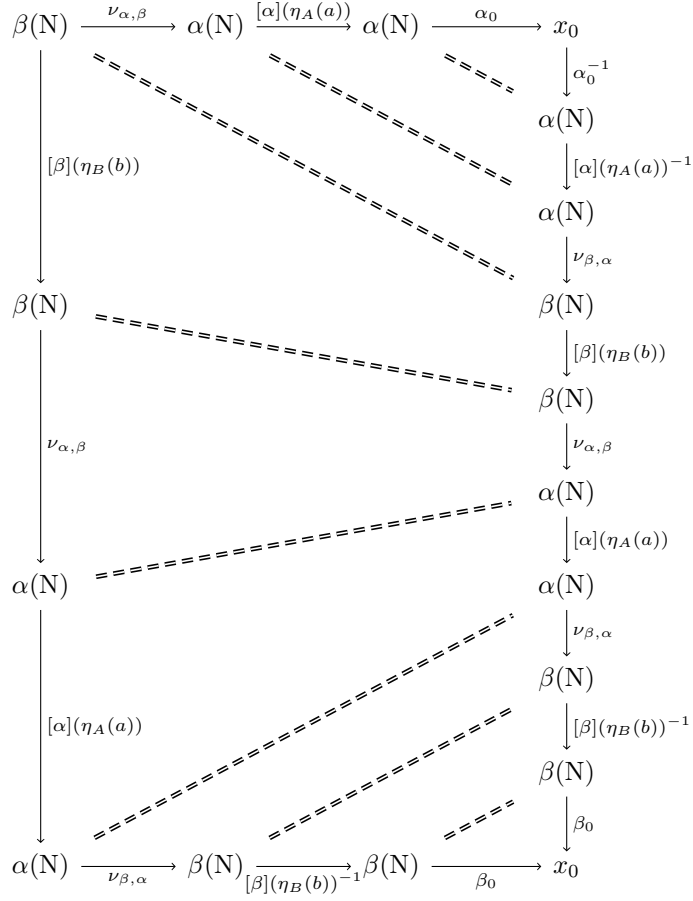
We fill this as indicated. This proves that
$\rho_\beta(\alpha)$ and $g$ are equal as mere functions.
We must still check that $H(\inl(a_0)) = \beta_0$.
This follows directly from $\ap{\alpha}(\eta_A(a_0)) = \refl{}$.

Specializing again to the case where $B \jdeq \Sn q$ is a sphere,
we have proved the following:
\begin{thm}\label{thm:sort-of-ehp}
  For any $\beta : \Sn{q+1} \ptdto X$, 
  there is a long exact sequence
  \[
    \begin{tikzcd}
      \cdots \ar[r] &
      \hgr{n+1}(\Sn{q+1} \to X, \beta)
      \ar[r] \ar[d, phantom, ""{coordinate, name=Z}] &
      \hgr{n+1}(X) \ar[dll, "\partial^{n,q}_\beta"{xshift=-10ex}, rounded corners,
      to path={ -- ([xshift=8.5ex]\tikztostart.center)
        |- (Z) 
        -| ([xshift=-13ex]\tikztotarget.center) [near start,swap] \tikztonodes -- (\tikztotarget)}] \\
      \hgr{n+q+1}(X) \ar[r] &
      \hgr{n}(\Sn{q+1} \to X,\beta) \ar[r] &
      \cdots
    \end{tikzcd}
  \]
  where the connecting homomorphisms are Whitehead products
  $\partial^{n,q}_\beta = [\blank,\settrunc{\phi^{q+1}_X(\beta)}]$ (where the
  bracket refers to the one defined in
  \cref{rem:whitehead-products-pis}).
\end{thm}

\section{Exploring the components of
  \texorpdfstring{$\Sn n = \Sn n$}{Sn=Sn}}
\label{sec:structure-components}

Having established that there are exactly two connected components of $(\Sn n = \Sn n)$, we want to examine the structure of each of these components.
The first observation we make is simple:

\begin{prop} \label{prop:general-case-conn-comp-equiv}
    For all $n \geq 1$, the two connected components of $(\Sn n = \Sn n)$ are equivalent.
\end{prop}
\begin{proof}
    For $n=1$ and $n=2$, this statement is given by the main result of \cref{sec:circle-case} and by \cref{prop:equiv-susp-comp}, respectively.
	For $n \geq 3$, it follows from \cref{prop:equiv-susp-comp} and \cref{cor:id-not-minus-id}.
\end{proof}

\cref{prop:general-case-conn-comp-equiv} means that we can restrict ourselves to the connected component of $\id$ (or $\refl{}$) in $(\Sn n = \Sn n)$. From now on, we use $\id$ as the implicit base point of $(\Sn n = \Sn n)$.
The rest of this subsection is devoted to calculating the fundamental group of this type, which also allows us to see that the equivalence $(\Sn 1 = \Sn 1) = (\Sn 1 + \Sn 1)$ does not generalize for $n > 1$.

The proof of the following lemma is in the appendix.
\begin{restatable}{lem}{lemmafgrSnptdtoSn}\label{lem:fgr-ptd-endomaps-Sn}
  For any $n\geq 1$, there is a group isomorphism:
  \begin{displaymath}
    \xi_n: \fgr(\Sn n \ptdto \Sn n, \id) \weq \hgr {n+1} (\Sn n)
  \end{displaymath}
\end{restatable}

\begin{rem}\label{rem:Sn-ptdto-Sn-equivalent-components}
  More generally, since $\Sn n \ptdto \Sn n$ is a wild group,
  given any $\alpha,\beta : \Sn n \ptdto \Sn n$,
  there is an equivalence between the corresponding components
  \[
    \conncomp{(\Sn n \ptdto \Sn n)}{\alpha} \weq
    \conncomp{(\Sn n \ptdto \Sn n)}{\beta}.
  \]
  This is given by $\gamma \mapsto \gamma \winv \alpha \wadd \beta$,
  where we write the group operation additively, cf.~\cref{prop:equiv-susp-comp}.
  This induces, for any $k \geq 1$, particular group isomorphisms of types
  \[
    \hgr k (\Sn n \ptdto \Sn n, \alpha) \weq
    \hgr k (\Sn n \ptdto \Sn n, \beta).
  \]
\end{rem}

We have now every tool needed to prove the result we
alluded to in \cref{sec:sphere}.%
\begin{restatable}{thm}{spsym}
  \label{thm:Sp-sym=Z/2}
  The fundamental group of the type of symmetries of $\Sp$ is the
  cyclic group of order $2$. That is, the following type has an
  element:
  \begin{equation*}
    \fgr(\Sp = \Sp, \id) = \ZZ \slash 2\ZZ
  \end{equation*}
  where both sides of the equality are considered as groups.
\end{restatable}
\begin{proof}[Sketch of proof]
  Using the long exact sequence of \cref{thm:sort-of-ehp} with
  $X \jdeq \Sp$, $n=q=1$ and $\beta \jdeq \id_\Sp$ , we get in
  particular an exact sequence:
  \begin{displaymath}
    \hgr 2(\Sp) \overset{[\blank, i_2]} \to \hgr 3(\Sp)
    \overset\kappa \to \fgr (\Sp \to \Sp, \id_\Sp) \to 0
  \end{displaymath}
  where $i_2$ generates the group $\hgr 2(\Sp)$ (equivalent to
  $\ZZ$). This shows that:
  \begin{displaymath}
    \fgr(\Sp = \Sp, \id) \weq \fgr (\Sp \to \Sp, \id_\Sp) \weq \hgr 3(\Sp) / \langle [i_2,i_2] \rangle
  \end{displaymath}
  The element $[i_2,i_2]$ is precisely that studied by Brunerie in its
  thesis: It generates a subgroup of index $2$ in the infinite cyclic
  group $\hgr 3(\Sp)$.
\end{proof}

The result generalizes to higher spheres. However, as it can be
expected from the use of Brunerie's number in the case of the sphere
$\Sp$, the results goes through for $n\geq 3$ for different
and actually simpler reasons, as shown now.
\begin{restatable}{thm}{thmfgrSntoSn}
  \label{thm:fund-grp-of-symmetries-based}%
  For $n \geq 3$, and any function $f : \Sn n \to \Sn n$, the
  fundamental group of the connected component of $f$ in the type
  $\Sn n \to \Sn n$ is the cyclic group of order $2$:
  \begin{equation*}
    \fgr(\Sn n \to \Sn n, f) = \ZZ \slash 2\ZZ
  \end{equation*}
  as groups.
\end{restatable}

\begin{proof}[Sketch of proof]
  Because $\fgr(\Sn n \to \Sn n, f) = \ZZ \slash 2\ZZ$ is a
  proposition and $\Sn n$ is connected, we can consider $f$ to be
  pointed. By considering the fiber sequence associated with the
  evaluation $\ev_\ast : \Sn n \to \Sn n
  \ptdto
  \Sn n$ at the north
  pole, we get a long exact sequence that allows to establish
  $\fgr(\Sn n \to \Sn n, f) \weq \fgr(\Sn n \ptdto \Sn n, f)$. By
  \cref{rem:Sn-ptdto-Sn-equivalent-components} and the wild adjunction
  $\susp{} \dashv \loopspace\null{}$, the latter is equivalent to
  $\hgr {n+1} (\Sn n)$, that is known to be $\ZZ/2\ZZ$.
\end{proof}

\begin{thm} \label{thm:fund-grp-of-symmetries} For $n \geq 3$, the
  fundamental group of the type of symmetries of $\Sn n$ is the cyclic
  group of order $2$:
    \begin{equation*}
    \fgr(\Sn n = \Sn n, \id) = \ZZ \slash 2\ZZ
    \end{equation*}
    as groups.
\end{thm}

\begin{proof}
  This is a direct corollary of
  \cref{thm:fund-grp-of-symmetries-based}, applied with
  $f\defequi \id_{\Sn n}$.
\end{proof}

Let us summarize the results of the current section, putting
\cref{thm:Sp-sym=Z/2,thm:higher-spheres-2,prop:general-case-conn-comp-equiv,thm:fund-grp-of-symmetries}
together:

\begin{thm}
  For $n \geq 2$, the type of symmetries of $\Sn n$ has two connected
  components. The two components are equivalent and both have
  fundamental group $\ZZ \slash 2\ZZ$. \qed
\end{thm}

\begin{rem}\label{rem:not-generalising}
  By \cref{thm:Sp-sym=Z/2,thm:fund-grp-of-symmetries}, for $n \geq 2$,
  the group $\fgr(\Sn n = \Sn n, \id)$ is non-trivial. At the same
  time, $\fgr(\Sn n)$ is trivial by \cite{HoTT}, and by
  $0$-connectedness of $\Sn n$, the group $\fgr(\Sn n + \Sn n, x)$ is
  trivial for any $x$.  Therefore, in contrast to the result for the
  case $n \equiv 1$ in \cref{sec:circle-case}, we have
  \begin{equation*}
    (\Sn n = \Sn n) \not= (\Sn n + \Sn n).
  \end{equation*}
\end{rem}

\section{Conclusion}
\label{sec:conclusions}

\newcommand{\topSp}{S^2}%
We have shown that the two connected components of $(\Sn n = \Sn n)$
are equivalent and have fundamental group $\ZZ/2\ZZ$ for $n\ge2$.
The only result the authors are aware of that gives
the complete homotopy type of these
components is a result proved in classical topology about the topological
$2$-sphere in~\cite[Sec.~5]{hansen}. Write $\topSp$ for the topological
$2$-sphere, and $\operatorname M(\topSp,\topSp)$ for the space of continuous
maps from $\topSp$ to $\topSp$ with the uniform topology:
\begin{thm}
  The connected component of the identity map in $\operatorname M(\topSp,\topSp)$
  is homotopy equivalent to $\operatorname{SO}(3)\times \Omega$, where $\Omega$
  is the universal covering space of the connected component of the constant
  loop in $\loopspace 2 (\topSp)$.
  \label{thm:hansen}
\end{thm}
This result can be stated as well in homotopy type theory,
since $\operatorname{SO}(3) \simeq \RR\mathrm{P}^3$,
and the real projective $3$-space $\RR\mathrm{P}^3$
can be defined as in~\cite{BuchholtzRijke2017}.
To prove the result, however, it might be necessary to
define the classifying type $\operatorname{BSO}(3)$,
which is itself another open problem.
(This problem is closely related to that of defining
the classifying type $\operatorname{BSU}(2)$,
where $\operatorname{SU}(2) \weq \Sn 3$.)
We leave the further investigations for future work.

Our proof of \cref{thm:sort-of-ehp} was inspired by
the similar result~\cite[Thm.~2.7]{lang1973},
the \emph{$\lambda$-component EHP sequence}.
As shown there, if we take $\beta\jdeq\id_{\Sn{q+1}}$,
we obtain an approximation to the classical EHP sequence~\cite{gwwhitehead1953}, valid in the range $n\le3q+1$.
It would be interesting to construct this in homotopy type theory,
as well as, of course, more modern refinements such as the various
EHP spectral sequences for each prime $p$, see~\cite[Sec.~1.5]{ravenel1986}
for a discussion in the classical setting.
For $p=2$, this is very much in reach using the James construction,
while for odd $p$, it would need Toda's fibrations~\cite{toda1962}.

Something special happens for the spheres that happen to be H-spaces,
namely $\Sn 0$, $\Sn 1$, $\Sn 3$, and $\Sn 7$.
Indeed, if $X$ is an H-space such that, say, left-multiplication is invertible,
meaning $\mu(x,\blank) : X \to X$ is invertible for all $x:X$,
then we get an equivalence
\[
  (X \to X) \simeq (X \ptdto X) \times X.
\]
This equivalence uses left-multiplication to adjust any function $f$
to a pointed function $x \mapsto \mu(f(x_0),\blank)^{-1}(f(x))$, together with
the image of the base point, $f(x_0)$.
Since this equivalence maps equivalences to pointed equivalences and vice versa,
it restricts to an equivalence $(X = X) \simeq (X \ptdweq X) \times X$.
Since $(\Sn 0\ptdweq \Sn 0) \simeq \bn1$ and
$(\Sn 1\ptdweq \Sn 1) \simeq \ZZ/2$, we recover the equivalences
\[
  (\Sn 0 = \Sn 0) \simeq \Sn 0,\qquad
  (\Sn 1 = \Sn 1) \simeq \ZZ/2 \times \Sn 1.
\]
From~\cite{BR18} we know that also $\Sn3$ is an H-space, and further,
that once we know that it's a loop space, then also $\Sn7$ is an H-space.
It also follows from our calculations that $\Sn2$ is not an H-space,
since all the components of $(\Sn2\ptdto\Sn2)\times\Sn2$ are equivalent,
but that's not the case for the components of $(\Sn2\to\Sn2)$.

We have focused on the self-identification types $\Sn n = \Sn n$, but
we could of course also look at the other components of the function
type $\Sn n \to \Sn n$.  For $n\ge3$, we don't get anything new by the
proof of \cref{thm:fund-grp-of-symmetries}, i.e.,
$\fgr(\Sn n\to\Sn n,f) \simeq \fgr(\Sn n\to\Sn n,\id) \simeq
\ZZ/2\ZZ$, for any $f$.  However, for $n=2$, from the proof of
\cref{thm:Sp-sym=Z/2}, the fundamental group of a component of
$\Sp \to \Sp$ depends on the corresponding degree.  Once we know that
the Whitehead product is bilinear (possibly up to a sign), as in
\cite[Prop.~3.4]{arkowitz1962}, we can conclude that
$\fgr(\Sp \to \Sp, f) = \ZZ/2d(f)\ZZ$.%

\bibliographystyle{plain}
\bibliography{../bib}

\appendix
\section{Proofs}

\lemmadeloopingcenter*

\begin{proof}
  Notice that $\conncomp {(BG=BG)} {\refl {BG}}$ is a connected
  groupoid, pointed at $(\refl {BG}, \trunc{\refl {\refl {BG}}})$. So,
  to exhibit a pointed equivalence as in the statement, we construct a
  pointed map $$z_G:\conncomp {(BG=BG)}{\refl {BG}} \ptdto BG$$ such
  that $\loopspace\null{(z_G)}$ is injective and has image
  $Z(G)$. Write $s:BG$ for the distinguished point of the delooping
  $BG$. The map $z_G$ is defined to be the restriction to the
  connected component at $\refl {BG}$ of the evaluation
  $\ev_s : (BG=BG) \to BG$ that sends an equality $x:BG=BG$ to the
  point $x(s):BG$ (where $x$ is seen as an equivalence $BG \weq
  BG$). Note that $z_G$ is pointed, trivially, by the path
  $\refl s : s = \ev_s(\refl{BG})$.

  Note that $\ev_y : (BG=BG) \to BG$ can be defined in the same way
  for any $y:BG$ instead of $s$. Then we can prove, by double
  path-inductions,
  \begin{displaymath}
    \ap x (q) \cdot \ap{\ev_s} (p) = \ap{\ev_y}(p) \cdot q
  \end{displaymath}
  for all $x : (BG = BG)$, all $p : \refl{BG} = x$, all $y:BG$, and
  all $q: s = y$. (Indeed, the equation holds for
  $p \jdeq \refl {\refl{BG}}$ and $q \jdeq \refl s$). In particular,
  the equation holds when $x \jdeq \refl{BG}$ and $y\jdeq s$, so that
  we have: for all $p : \refl{BG} = \refl{BG}$ and all $g : s = s$,
  $g \cdot \ap{\ev_{s}} (p) = \ap{\ev_{s}}(p) \cdot g$. By restriction
  to the subtype $\conncomp {(BG=BG)} {\refl {BG}}$, we get: for all
  $h : \loopspace\null{\conncomp {(BG=BG)} {\refl {BG}}}$, and all
  $g : G$, $g \cdot \ap{z_G} (h) = \ap{z_G}(h) \cdot g$. Because $z_G$
  is pointed by $\refl s$, path algebra shows that
  $\loopspace\null{(z_G)} = \ap {z_G}$. Hence, for any
  $h : \loopspace\null{\conncomp {(BG=BG)} {\refl {BG}}}$,
  $\loopspace\null{(z_G)}(h)$ lies in the center of $G$.

  Conversely, we must show that any element of the center is in the
  image of by $\loopspace\null{(z_G)}$. Take $g$ in the center of $G$,
  and construct an element of
  $\loopspace\null{\conncomp {(BG=BG)} {\refl {BG}}}$ as
  follows. Define $\hat g : \refl {BG} = \refl {BG}$ through
  univalence by giving an equality of type $\id_{BG} = \id_{BG}$, that
  is, under function extensionality, by giving a homotopy of type
  $\prod_{y:BG}y=y$. We will obtain such a homotopy by taking, for all
  $y:BG$, the first component of a center of contraction of
  $\sum_{q : y=y}\prod_{q':s=y}(q'g=qq')$. The contractibility of
  $\sum_{q : y=y}\prod_{q':s=y}(q'g=qq')$ is a proposition, so we can
  use the connectedness of $BG$ and only prove it for $y\jdeq s$. But
  because $g$ commutes with all elements of $G$,
  \begin{displaymath}
    \sum_{q : G}\prod_{h:G}(hg=qh) \weq \sum_{q : G} G \to (g = q)
    \weq \sum_{q : G}( g = q) \weq 1.
  \end{displaymath}
  We obtain in that way $\hat g : \refl {BG} = \refl {BG}$ and we
  consider the element
  $(\hat g,!) : \loopspace\null{\conncomp {(BG=BG)} {\refl
      {BG}}}$. Its image though $\loopspace\null{(z_G)}$ is
  $\hat g(s)$ when $\hat g$ is seen as an homotopy
  $\prod_{y:BG}y=y$. And by definition, $\hat g (s)$ is the first
  component of the center of contraction of
  $\sum_{q : G}\prod_{h:G}(hg=qh)$. Tracking back the equivalences
  above, that is exactly $g$. We thus have proven that
  $\loopspace\null{(z_G)}$ has image $Z(G)$.

  Lastly, we have to show that $\loopspace\null{(z_G)}$ is
  injective. Equivalently, we want to prove that all fibers of $z_G$
  are sets. This is a proposition, so it boils down to proving that
  the fiber at $s$ is a set. This fiber is the type of elements
  $x : BG=BG$, merely equal to $\refl{BG}$, together with an equality
  from $s$ to $\ap{\ev_s}(x)$. Through univalence, this is thus the
  type of pointed equivalences from $BG$ to itself. But this is
  equivalent to the type of group automorphisms of $G$, which is a
  set.
\end{proof}

\lemmaouterautascomponents*

\begin{proof}
  Recall the morphism of groups $\inn : G \to \Aut(G)$ that maps an
  element $g\in G$ to the inner automorphism $x \mapsto g x
  g^{-1}$. By univalence, the delooping of $\Aut(G)$ can be described
  as the connected component of $BG$ in $\UUptd$. In particular, under
  this identification, $\inn$ is simply $\loopspace\null(B\inn)$ where
  $B\inn : BG \ptdto \conncomp \UUptd {BG}$ is the function mapping
  each $y:BG$ to the type $BG$ itself but pointed at $y$ instead of
  the distinguished point $s$ of $BG$. The path pointing $B\inn$ is
  simply $\refl{BG}$. Using the fact stated above the lemma, we find a
  pointed equivalence $\setTrunc{\inv{B\inn}(BG)} \weq \Out(G)$ by
  pointing the type on the left at $\settrunc{(s, \refl{BG})}$. But
  the fiber $\inv{B\inn}(BG)$ is by definition
  $\sum_{y:BG}(BG,y) =_\ast BG$, which is equivalent to $BG =
  BG$. Moreover, through this last equivalence, $\refl{BG}: BG = BG$
  corresponds to $(s, \refl{BG}) : \inv{B\inn}(BG)$. Truncating the
  equivalence, we get a pointed equivalence
  $\setTrunc{BG = BG} \ptdweq \setTrunc{\inv{B\inn}(BG)}$, and thus by
  composing with the first pointed equivalence, we get
  $\setTrunc{BG = BG} \ptdweq \Out(G)$ as wanted.
\end{proof}

\lemminusidequivalence*

\begin{proof}
  More generally, the same holds for the reflection $-\id_{\susp X} : \susp X\to\susp X$
  on any suspension, so we prove it in this generality.
  We produce by induction an element of the type $\prod_{z:\susp X}T(z)$,
  where
  $T(z)\defequi (z= (-\id_{\susp X} \circ -\id_{\susp X})(z))$.
  By definition of $-\id_{\susp X}$,
  $T(\north) \jdeq (\north=\north)$
  and $T(\south) \jdeq (\south=\south)$,
  so we take $\refl \north:T(\north)$
  and $\refl \south:T(\south)$.
  To complete the induction, we need to provide an element
  of type $\prod_{x:X} \pathover{\refl \north} T {\mrd x} {\refl \south}$.
  Transporting over a meridian in the family $T$ is conjugation by the
  meridian: indeed, the transport over any path $p:x=x'$ in $T$ is given by
  $q\mapsto \ap{-\id_{\susp X} \circ -\id_{\susp X}}(p) \cdot q \cdot \inv p$, and
  \begin{align*}
    \ap{-\id_{\susp X} \circ -\id_{\susp X}} (\mrd(x))
     &= \ap{-\id_{\susp X}}(\inv {\mrd(x)})
    \\ &= \inv {\bigl( \ap{-\id_{\susp X}}(\mrd(x)) \bigr)}
    \\ &= \inv { \left( \inv{\mrd(x)} \right) }
    \\ &= \mrd(x).
  \end{align*}
  Hence $\pathover{\refl \north} T {\mrd(x)} {\refl \south}$
  is equivalent to
  $\mrd(x)\refl \north \inv {\mrd(x)} = \refl\south$, which is indeed inhabited for any
  $x:\Sc$ by simple path algebra.
\end{proof}

\lemtauretractionunitcircle* 
\begin{proof}
  For all $z\in\Sc$, we can calculate:
  \begin{align*}
    \tau(\eta_{\Sc}(z))
    &= \ap{\hopf}(\inv{\mrd(\base)} \cdot \mrd(z))(\base) \\
    &= (\iota_{\base}^{-1} \circ \iota_z) (\base)
      = \id_{\Sc}^{-1}(\iota_z (\base))= z
  \end{align*}
  It remains to verify that $\ap\tau (\eta_0)$ coincides with the path above
  instantiated at $z\jdeq \base$.
  In fact, $\tau$ factors through an \emph{adjoint} equivalence
  $\bar\tau : \Trunc{\loopspace\null\Sp}_1 \simeq_* \Sc$,
  with inverse $\bar\eta_\Sc = \trunc{\blank}_1 \circ \eta_\Sc$.
  The identification $\bar\eta_\Sc \circ \bar\tau = \id_{\Trunc{\loopspace\null\Sp}_1}$
  defined by the encode-decode method is more easily seen to be pointed,
  but then the identification $\bar\tau \circ \bar\eta_\Sc = \tau \circ \eta_\Sc = \id_\Sc$ is pointed as well.
  The full details of this argument are formalized in~\cite{hott-agda}.
\end{proof}

\leminvetascequalsetascminusid* 
\begin{proof}
  We construct the element $\kappa$ by induction on $\Sc$ by first defining
  \begin{displaymath}
    \kappa_{\base} : \inv{(\inv{\mrd(\base)}\cdot\mrd(\base))} = \inv{\mrd(\base)}\cdot\mrd(\base)
  \end{displaymath}
  by simple path algebra.

  It remains to find an element $\kappa_{\Sloop}: \pathover{\kappa_{\base}}{}{\Sloop}{\kappa_{\base}}$, which amounts to an element of 
  \begin{equation*}
    \inv{\ap{\eta_\Sc}(\Sloop)} \cdot \kappa_{\base} = \kappa_{\base} \cdot \ap{\inv\blank}(\ap{\eta_\Sc}(\Sloop))
  \end{equation*}
  Consider the following paths given by path algebra for each $x:\Sc$:
  \begin{align*}
    &\lambda_x: \inv{(\inv{\mrd(x)}\cdot\mrd(\base))} = \eta_\Sc(x)\\
    &\mu_x: \inv{\eta_\Sc(x)} = \inv{\mrd(x)}\mrd(\base)
  \end{align*}
  In particular, the types $\lambda_{\base} = \kappa_{\base}$, $\mu_{\base} = \kappa_{\base}$, and $\lambda_{\base} = \mu_{\base}$ have elements.
  Moreover, by path induction on $p:\base = x$, one can construct an element
  \begin{equation*}
    \xi_p : \inv{\ap{\eta_\Sc}(p)} \cdot \lambda_x = \mu_x \cdot \ap{\inv\blank}(\ap{\eta_\Sc}(p))
  \end{equation*}
  The element $\xi_{\Sloop}$ gives the wanted $\kappa_{\Sloop}$ by transport.
\end{proof}

\propdegreemonoidmorphism* 
\begin{proof}
  First, let us prove that $\id_\Sp$, pointed by $\refl\north$, has degree $1$.
  This is easy because $\hgr 2(\id_\Sp) = \id_{\hgr 2(\Sp)}$ so that $d(\id_\Sp)
  = \zeta(\inv\zeta(1)) = 1$.

  Now, let us prove that $d(g\circ f) = d(g)\times d(f)$ for any $f,g:\Sp
  \ptdto\Sp$. This again comes mainly from the functoriality of $\hgr 2$
  (\cite[after Lem.~7.3.3 and before Def.~8.4.2]{HoTT}),
  meaning that $\hgr 2(g\circ f) = \hgr 2(g)\circ\hgr 2(f)$
  holds. Hence:
  \begin{align*}
    d(g\circ f) &= \zeta\left( \hgr 2(g\circ f)\left( \inv\zeta(1) \right) \right)
    \\
    &= \zeta\left( \hgr 2(g) \left( \hgr 2(f)\left( \inv\zeta(1) \right)
    \right) \right)
    \\
    &= \zeta\left( \hgr 2(g) \left( \inv\zeta \left( \zeta \left( \hgr 2(f)\left( \inv\zeta(1) \right)
    \right) \right) \right) \right)
    \\
    &= \zeta \left( \hgr 2(g) \left( \inv\zeta \left(d(f)\right) \right) \right)
  \end{align*}
  Here, we can use the fact that $\zeta$ is not just any equivalence
  but actually a group isomorphism. Because $\hgr 2 (g)$ is a
  homomorphism of groups, the composition
  $\zeta \circ \hgr 2 (g) \circ \inv \zeta : \ZZ \to \ZZ$ also
  is. Hence, for any $n:\ZZ$, one gets
  $(\zeta \circ \hgr 2 (g) \circ \inv \zeta)(n) = {(\zeta \circ \hgr 2
    (g) \circ \inv \zeta)(1)}\times n$. We can then conclude:
  \begin{displaymath}
    d(g\circ f)
    = \zeta \left( \hgr 2(g) \left( \inv\zeta(1)\right) \right)\times d(f)
    = d(g)\times d(f)\qedhere
  \end{displaymath}
\end{proof}

\cordegreeequivalences* 
\begin{proof}
  Given an equivalence $\varphi:\Sp \to \Sp$, pointed by $p:\varphi(\north) = \north$,
  any inverse $\psi$ of $\varphi$ is also pointed by the following path $q$:
  \begin{displaymath}
    \psi(\north) \overset {\ap\psi(\inv p)} = \psi\varphi(\north) \overset {\alpha_\north}  = \north
  \end{displaymath}
  where $\alpha: \psi\varphi = \id$ is a witness of $\psi$ being a left inverse
  for $\varphi$. In particular, $(\psi, q) \circ (\varphi, p)$ is an
  equivalence whose first component is equal to $\id_\Sp$. In determining the
  degree of this composite equivalence, \cref{lem:deg-independent-path} ensures
  that the path is irrelevant, and because $d(\id_\Sp, \refl\north) = 1$, we can
  conclude that $d((\psi, q) \circ (\varphi, p)) = 1$ also holds. The previous
  result then proves that $d(\varphi, p)$ is a divisor of $1$ in $\ZZ$, which
  is either $1$ or $-1$ by decidability of the equality in $\ZZ$.
\end{proof}

\lemmaconnapf* 

\begin{proof}
Let $a_1, a_2 :A$, and $p : f(a_1) = f(a_2)$.
    We have to prove that $\higherTrunc{k-1}{\inv{\ap{f}}(p)}$
    is contractible.
    From \cite[proof of Lem.~7.6.2]{HoTT} we get that
    the fiber $\inv{\ap{f}}(p)$ is equivalent to
    the path type $(a_1,p) = (a_2, \refl {f(a_2)})$ in
    the fiber $\inv{f}(f(a_2))$. 
    Moreover, by \cite[Thm.~7.3.12]{HoTT}, 
    $\higherTrunc {k-1} {(a_1, p) = (a_2, \refl {f(a_2)})}\weq
    \highertrunc k {(a_1, p)}= \highertrunc k {(a_2, \refl {f(a_2)})}$.
    The latter path type is contractible if
    $\higherTrunc k {f^{-1}(f(a_2))}$ is contractible,
    which follows from the assumption of the lemma.
\end{proof}

\lemmainvertables*

\begin{proof}
    For giving maps in both directions we use set truncation elimination.
Being an equivalence and being invertible are propositions and we
denote proofs if they exist simply by !.
For each $(f,!):A \weq A$ we map  $\settrunc{(f,!)}$ to $(\settrunc{f},!)$;
obviously, $\settrunc{f}$ is invertible with inverse $\settrunc{\inv f}$
if $f$ is an equivalence. For the converse,
if $x:\setTrunc{A \to A}$ is invertible, then
we have the unique inverse $\inv x$. We may assume
that $x \jdeq \settrunc f$ for some function $f:A \to A$ and
$\inv x \jdeq \settrunc g$ for some $g:A \to A$.
    From the inverse law in the monoid $\setTrunc{A \to A}$, 
    using \cite[Thm.~7.3.12]{HoTT}, one
    derives that both $fg$ and $gf$ are merely equal to $\id_{A}$. To prove
    the proposition that $f$ is an equivalence, one can then assume actual
    witnesses of $fg=\id_{A}$ and $gf=\id_{A}$. Then $g$ is a
    pseudo-inverse for $f$. Hence the map in the other direction
    is $(\settrunc f,!) \mapsto \settrunc {(f,!)}$.
Clearly the maps in both directions are pseudo-inverses. 
\end{proof}

\lemmaUMPjoinvsloop*

\begin{proof}
  Let $x_0$ be the base point of $X$.
  Given $f: A \ptdto (B \ptdto \loopspace{}X)$, construct $\bar f:A*B \ptdto X$
  by induction:
  \begin{equation*}
    \begin{aligned}
      \bar f(\inl(a)) &\defequi x_0 \quad \text{for $a:A$} \\
      \bar f(\inr(b)) &\defequi x_0 \quad \text{for $b:B$} \\
      \ap{\bar f}(\glue(a,b)) &\defis f(a)(b) \quad \text{for $a:A,b:B$}
    \end{aligned}
  \end{equation*}
  The map $\bar f$ is trivially pointed.

  This construction $f\mapsto \bar f$ admits an inverse. Let $a_0$ and $b_0$ be the base points of $A$ and $B$, respectively.
  Now, for any $a:A$ and $b:B$, one has an element $\tau_{a,b}$
  of $\loopspace{}(A*B)$ constructed as the following composition of paths:
  \begin{equation}
    \begin{tikzcd}[column sep=large]
      \inl(a_0) \rar["{\glue(a_0,b_0)}"] & \inr(b_0) \rar["\inv{\glue(a,b_0)}"]
      & \inl(a) \rar["{\glue(a,b)}"] & \inr(b) 
      \arrow[lll, bend left, "\inv{\glue(a_0,b)}"]
    \end{tikzcd}
    \label{eq:def-tau-ab}
  \end{equation}
  Then, to any $g:A*B \ptdto X$, one can map the function $\hat g: a \mapsto
  (b\mapsto \loopspace\null(g)(\tau_{a,b}))$.
  The map $\hat g$ is a pointed, as
  $\tau_{a_0,b} = \tau_{a,b_0} = \refl{\inl(a_0)}$ for all $a:A$ and $b:B$
  by path algebra.
  The construction $g
  \mapsto \hat g$ provides an inverse to $f \mapsto \bar f$.
  A proof of this has been formalized in cubical Agda~\cite{joinloopadj}.
  There, it's also checked that this equivalence arises from
  a wild adjunction, cf.~\cref{prop:join-loop-adjunction}.
\end{proof}

\lemmafgrSnptdtoSn*

\begin{proof}
  Recall the equivalence $\phi_{\Sn n}^n : (\Sn n \ptdto \Sn n) \to \loopspace
  {n} (\Sn n)$, defined in \cref{def:sphere-to-eq-loops}. Note that
  this is not a pointed equivalence. Indeed, $\phi_{\Sc}^1$ maps $(\id_\Sc,
  \refl\base)$ to $\Sloop: \loopspace\null(\Sc)$ and from there, one can prove by
  induction that $\phi_{\Sn n}^n$ maps $(\id_{\Sn n}, \refl\north)$ to a point in
  $\loopspace n (\Sn n)$ which is mapped to $\pm 1:\ZZ$ by the set truncation
  $\loopspace n (\Sn n) \to \hgr n (\Sn n) \weq \ZZ$. However, the distinguished
  point of $\loopspace n (\Sn n)$ is the iterated $\refl{}$, which is sent to
  $0:\ZZ$ by this set truncation.

  Fortunately, there is an equivalence $\psi_n:\loopspace n (\Sn n) \weq
  \loopspace n(\Sn n)$ defined as follows:
  \begin{equation*}
    \psi_n (\alpha) \defequi \inv{\phi_{\Sn n}^n(\id_{\Sn n},\refl\north)}\cdot \alpha
  \end{equation*}
  This makes the composite $\psi_n \circ \phi_{\Sn n}^n$ pointed by path algebra.
  The wanted equivalence is then:
  \begin{equation*}
    \xi_n\defequi \fgr(\psi_n \circ \phi_{\Sn n}^n) :
    \fgr(\Sn n \ptdto \Sn n) \weq \hgr {n+1} (\Sn n)
  \end{equation*}
  once we identify $\hgr {n+1}(\Sn n)$ with $\setTrunc{\loopspace {}(\loopspace
  n (\Sn n))}$.
\end{proof}

\spsym*

\begin{proof}
  We follow the proof (in classical topology) of \cite{gwwhitehead}.
  Let us consider the long exact sequence given by \cref{thm:sort-of-ehp} when
  specialized to $X \jdeq \Sp$, $n=q=1$ and $\beta\jdeq (\id_\Sp, \refl\north): \Sp
  \ptdto \Sp$. Using $\phi_\Sp^2$ from \cref{def:sphere-to-eq-loops},
  the element in $\hgr 2(\Sp)$ represented by $(\id_\Sp, \refl\north)$
  is:
  \begin{equation}
    i_2 \defequi \settrunc{\loopspace{}(\eta_\Sc)(\Sloop)}
    \label{eq:def-i2-gen-pi2S2}
  \end{equation}
  Recall that this element $i_2$ generates the group $\hgr 2(\Sp)$, which is
  isomorphic to $\ZZ$.

  We then have an exact sequence:
  \begin{equation}
    \hgr 2(\Sp) \overset{[\blank, i_2]} \to \hgr 3(\Sp)
    \overset\kappa \to \fgr (\Sp \to \Sp, \id_\Sp) \to \fgr (\Sp)
    \label{eq:exact-seq-compute-pi1-S2eqS2}
  \end{equation}
  We consider the type $\Sp \to \Sp$ as pointed at the element $\id_\Sp$. Hence,
  we write only $\fgr(\Sp \to \Sp)$ instead of $\fgr(\Sp\to\Sp,\id_\Sp)$.
  The sphere $\Sp$ is simply connected, that is $\fgr(\Sp) = 0$. By exactness,
  it means that:
  \begin{equation}
    \fgr(\Sp \to \Sp) \weq \im(\kappa) \weq \hgr 3(\Sp) / \ker(\kappa)
    \weq \hgr 3 (\Sp) / \langle [i_2,i_2] \rangle
    \label{eq:fgr-SpeqSp-as-quotient}
  \end{equation}

  Recall from \cref{lem:fgr-ptd-endomaps-Sn} that there is a group isomorphism
  $\xi_2: \hgr 3 (\Sp) \weq \ZZ$. Hence, one has:
  \begin{equation}
    \fgr (\Sp = \Sp) \weq \fgr (\Sp \to \Sp) \weq \ZZ / \xi_2([i_2,i_2])\ZZ
    \label{eq:fgr-SpeqSp-ZovernZ}
  \end{equation}
  We now invoke the main result from \cite{brunerie:thesis} to conclude that
  $\fgr (\Sp = \Sp) = \ZZ/2\ZZ$ (as groups).
\end{proof}

\thmfgrSntoSn*

\begin{proof}
  As the group $\ZZ/2\ZZ$ has no non-trivial symmetries, the target
  type $\fgr(\Sn n \to \Sn n, f) = \ZZ \slash 2\ZZ$ is a
  proposition. So, as the sphere $\Sn n$ is connected, we can suppose
  without loss of generality that $f$ is pointed by a path
  $f_0: f(N) = N$. In the rest of the proof, we make the abuse of
  writing $f$ for both the pointed and unpointed map as the context
  allows to differentiate.

  Consider the following fiber sequence:
  \begin{displaymath}
    (\Sn n \ptdto \Sn n, f)\stackrel \iota \longrightarrow_\ast
    (\Sn n \to \Sn n, f) \stackrel {\ev_\ast} \longrightarrow_\ast \Sn n
  \end{displaymath}
  where $(\Sn n \ptdto \Sn n, f)$ is the type $\Sn n \ptdto \Sn n$
  pointed at $f$, $(\Sn n \to \Sn n, f)$ is the type $\Sn n \to \Sn n$
  pointed at $f$, $\iota$ is the map forgetting the pointing path, and
  $\ev_\ast$ is the evaluation at the point $N$ of $\Sn n$. By
  \cite[Thm.~8.4.6]{HoTT}, this induces a long exact sequence of
  groups: {\small
  \begin{displaymath}
    \dots \to \hgr{k+1} (\Sn n)
    \to \hgr k(\Sn n \ptdto \Sn n, f) \to \hgr k(\Sn n \to \Sn n, f) \to \hgr k(\Sn n)
    \to \dots
  \end{displaymath}}
  Hence for every $1\leq k < n-1$, this long sequence contains the
  following short exact sequence: {\small
  \begin{displaymath}
    0 = \hgr {k+1}(\Sn n)
    \to \hgr k(\Sn n \ptdto \Sn n, f) \to \hgr k(\Sn n \to \Sn n, f)
    \to \hgr k(\Sn n) = 0
  \end{displaymath}}
  In other words, for every $1\leq k < n-1$, one has
  \begin{displaymath}
    \hgr k(\Sn n\to \Sn n, f) \weq \hgr k(\Sn n\ptdto\Sn n, f)
    \weq \hgr {k}(\loopspace n \Sn n) \weq \hgr {n+k} (\Sn n)
  \end{displaymath}
  The group isomorphism in the middle is given by
  \cref{rem:Sn-ptdto-Sn-equivalent-components} as follows: if we write
  $r_n$ for the chosen point of $\loopspace n \Sn n$ (that is the
  iterated $\refl{}$ on the north pole $N$ of $\Sn n$), then the
  remark allows to derive an isomorphism
  $\hgr k(\Sn n\ptdto\Sn n, f) \weq \hgr {k}(\Sn n \ptdto \Sn n,
  \inv{\left(\phi_{\Sn n}^n\right)} (r_n))$, that can be composed with
  the isomorphism
  $\hgr {k}(\Sn n \ptdto \Sn n, \inv{\left(\phi_{\Sn n}^n\right)}
  (r_n)) \weq \hgr k (\loopspace n\Sn n)$ induced by the equivalence
  $\phi_{\Sn n}^n : (\Sn n \ptdto \Sn n) \weq \loopspace n \Sn n$.
  
  In particular for $n\geq 3$, $k\defequi 1$ enters the condition, and one gets:
  \begin{displaymath}
    \fgr (\Sn n \to \Sn n, f) \weq \hgr {n+1} (\Sn n) \weq \ZZ \slash 2\ZZ\qedhere
  \end{displaymath}
\end{proof}
\section{Wild categories}

\def\hom(#1,#2){#1 \rightsquigarrow #2}
The types $\UUptd$ and $\UU$ certainly do not form categories in the usual
sense (the intended types of morphisms $A \ptdto B$ and $A\to B$ between two
objects $A,B$ are not necessarily sets), but some constructions on $\UUptd$ and $\UU$ are
reminiscent of functors. This motivates the few following definitions.

\begin{defi}
  A \emph{wild category in} $\UU$ is a dependent tuple $(\cat C, \hom( , ), \circ, \id, \alpha,
  \iota)$ where:
  \begin{align*}
    \cat C&: \UU \\
    \hom( , )&: \cat C \times \cat C \to \UU \\
    \circ&: \prod_{a,b,c:\cat C} \hom(b,c) \times \hom(a,b) \to \hom(a,c)
    \quad (\text{notation}\ g\circ f \defequi \circ(a,b,c)(g,f)) \\
    \id&: \prod_{a:\cat C} \hom(a,a) \quad (\text{notation}\ \id_a \defequi \id(a)) \\
    \alpha&: \prod_{a,b,c,d:\cat C} \prod_{f:\hom(a,b)} \prod_{g:\hom(b,c)} \prod_{h:\hom(c,d)}
    h\circ (g \circ f) = (h \circ g) \circ f \\
    \iota&: \prod_{a,b:\cat C} \prod_{f:\hom(a,b)} (f\circ \id_a = f) \times (f = \id_b \circ f)
  \end{align*}
  \label{defn:wild-cat}
\end{defi}
One makes the abuse to denote a wild category by only its carrier type $\cat C$ when
all the remaining data are clear from context.

\begin{rem}
  To be fully rigorous, one must say a word about the levels of universe one
  allows in the definition of wild categories. Assuming we have a cumulative
  hierarchy of universes $\UU_0 : \UU_1 : \ldots$, we choose to consider
  locally small wild categories, by which are meant wild categories such that
  $\cat C:\UU_{i+1}$ and $\hom(,): \cat C \times \cat C \to \UU_i$ for some
  $i\geq 0$.
  \label{rem:universes-locally-small-cats}
\end{rem}

\begin{exa}
  The type $\UU_i$ together with function types, identity functions and the
	usual composition, is a wild category in $\UU_{i+1}$, for any $i$. (We will henceforth ignore universe levels and just write $\UU$.) The elements $\alpha$ and $\iota$ are
  given by function extensionality.

  Similarly, the type $\UUptd$ together with $\ptdto$, identity functions
  pointed by $\refl{}$ paths, and composition of pointed functions, is a wild
  category. Again, the elements $\alpha$ and $\iota$ are given by function
  extensionality, completed by path algebra.
\end{exa}

\begin{defi}
  Let $\cat C,\cat D$ be two wild categories.
  A \emph{wild functor} from $\cat C$ to $\cat D$ is a dependent $4$-tuple $(F_0,F_1,c,u)$ where:
  \begin{align*}
    F_0 &:\cat C \to \cat D,\\
    F_1 &:\prod_{a,b:\cat C} \hom(a,b) \to \hom(F_0(a), F_0(b))\\
    c &: \prod_{a,b,c:\cat C}\prod_{f:\hom(a,b)}\prod_{g:\hom(b,c)}F_1(g\circ f) = F_1(g) \circ F_1(f)\\
    u &: \prod_{a:\cat C} F_1(\id_a) = \id_{F_0(a)}
  \end{align*}
  \label{def:wild-functor}
\end{defi}

As for non-wild functors, we usually write both $F_0$ and $ F_1$ only as $F$.
The only relevant fact about $c$ and $u$ is that they exist, even though their
types are not propositions. Therefore, we will often (abusively) denote a wild
functor by its first two components only.

\begin{prop}
  Given wild functors $F:\cat C \to \cat D$ and $G:\cat D \to \cat E$, there is a composite wild functor $G\circ F :\cat C \to \cat E$ with first components:
  \begin{align*}
    (G\circ F)_0 &\defequi G_0\circ F_0, \\
    (G\circ F)_1 &\defequi (a,b) \mapsto G_1(F_0(a),F_0(b)) \circ F_1(a,b)
  \end{align*}
  \label{prop:composition-wild-functors}
\end{prop}
\begin{proof}
  Denote $c_F, u_F$ and $c_G,u_G$ the witnesses of functoriality for $F$ and
  $G$ respectively. Then one gets:
  \begin{align*}
    c_{G\circ F}(a,b,c)(f,g) &\defequi c_G(F_0(a),F_0(b),F_0(c),F_1(f),F_1(g)) \\
    & \phantom{\defequi} \cdot \ap{G_1} ( c_F(a,b,c,f,g) ) \\
    u_{G\circ F}(a) &\defequi u_G(F_0(a)) \cdot \ap {G_1} (u_F(a))\qedhere
  \end{align*}
\end{proof}

\begin{exa}~\label{ex:loop-sus-wild-functors}%
  \begin{enumerate}
    \item There is a wild functor $\loopspace \null$ from $\UUptd$ to itself
      which maps a pointed type $A$ (pointed at $a$) to $\loopspace \null A
      \defeq (a=a$) (pointed at $\refl a$), and maps a pointed function $f: A
      \ptdto B$ to the pointed function $\loopspace \null (f): \loopspace \null
      A \ptdto \loopspace \null B$ defined as follows:
      \begin{displaymath}
        \loopspace \null (f) (p) \defequi f_0 \cdot \ap f (p) \cdot \inv{f_0}
      \end{displaymath}
      where $f_0: f(a) = b$ is the path pointing $f$.
      The map $\loopspace \null (f)$ is itself pointed by a path obtained by
      using path algebra as follows:
      \begin{displaymath}
        \loopspace\null (f)(\refl a) \jdeq  f_0 \cdot \ap f (\refl a) \cdot \inv{f_0}
        = f_0 \cdot \refl {f(a)} \cdot \inv{f_0}
        = \refl b.
      \end{displaymath}
      A careful exposition of the witness created this way can be found in
      \cite[\githubpath core/lib/types/LoopSpace.agda]{hott-agda}. It can also
      be found the witnesses $c$ and $u$ justifying that $\loopspace \null$ is
      a wild functor.
    \item There is a wild functor $\susp \null$ from $\UU$ to $\UUptd$ which
      maps a type $A$ to $\susp A$ (pointed at the pole $\north$), and maps a
      function $f:A\to B$ to the pointed function $\susp (f) : \susp A \ptdto
      \susp B$ defined by induction as follows:
      \begin{displaymath}
        \susp (f) (\north) \defequi \north,\quad
        \susp (f) (\south) \defequi \north,\quad
        \ap{\susp (f)} \circ \mrd \defis \mrd \circ f.
      \end{displaymath}
      Note that $\susp (f)$ is pointed by $\refl{\north}$ as it maps $\north$ to $\north$
      by definition. The witnesses $u$ and $c$ of \cref{def:wild-functor} are
      defined through easy inductions on the suspension, and a
      careful exposition can be found in
      \cite[\githubpath core/lib/types/Suspension.agda]{hott-agda}.
    \item The join operation ${\blank} * B$ from $\UUptd$ to itself,
      defined in~\cref{sec:whitehead-interlude},
      has the structure of a wild functor, mapping $f : A \ptdto A'$
      to the pointed function $f * B : A * B \ptdto A' * B$
      defined by induction as follows:
      \begin{align*}
        (f * B)(\inl(a)) &\defeq \inl(f(a)),\qquad
        (f * B)(\inr(b)) \defeq \inr(b),\\
        \ap{f*B}(\glue(a,b)) &\defis \glue(f(a),b).
      \end{align*}
      Note that $f * B$ is pointed by
      $\ap{\inl}f_0 : \inl(f(a_0)) = \inl(a_0')$,
      where $a_0$ and $a_0'$ are the base points of $A$ and $A'$,
      and $f_0 : f(a_0) = a_0'$ witnesses that $f$ is pointed.

      A formalization of $u$ and $c$ of \cref{def:wild-functor}
      in cubical Agda is in \cite{joinloopadj}.
    \item There is a wild functor $U:\UUptd \to \UU$ that maps a pointed type
      $(A,a_0)$ to $A$ and a pointed maps $(f,f_0)$ to $f$. The witnesses $c$
      and $u$ are both given by reflexivity. In practice, the application of $U$
      is left implicit: we write for example $\susp(X)$ for $\susp (U (X))$
      when $X$ is a pointed type. In particular, depending on the context, the
      wild functor $\susp$ can be considered to have domain $\UUptd$.
    \item There is a wild functor $\setTrunc \blank$ from $\UUptd$ to $\UU$
      which maps a pointed type $(A,a_0)$ to $\setTrunc{A}$ and a pointed map
      $(f,f_0)$ to the map $\setTrunc f$ defined by induction as $\settrunc x
      \mapsto \settrunc{f(x)}$. The witnesses $c$ and $u$ are defined
      by using well-known inhabitants of the following two types, respectively:
      \begin{gather*}
        \setTrunc g \circ \setTrunc f \circ \settrunc \blank = \settrunc \blank \circ (g\circ f),
        \\
        \setTrunc{\id} \circ \settrunc \blank = \settrunc \blank \circ \id
      \end{gather*}

    \item Building on the previous examples and
      \cref{prop:composition-wild-functors}, there is for each $n:\NN$, a wild
      functor $\hgr n$ from $\UUptd$ to $\UU$ which acts on objects and maps as
      $\setTrunc{\loopspace n (\blank)}$. The witnesses $c$ and $u$ are given
      by successive transport and composition of the same witnesses for
      $\loopspace\null$ and $\setTrunc \blank$,
      as explained in \cref{prop:composition-wild-functors}.
  \end{enumerate}
\end{exa}

\begin{defi}
  Let be given a wild functor $L$ from $\cat C$ to $\cat D$, and a wild functor
  $R$ from $\cat D$ to $\cat C$. A \emph{wild adjunction} of type $L\leftadjto
  R$ consists of the data of two dependent functions, the unit and the counit:
  \begin{displaymath}
    \eta : \prod_{a:\cat C} \hom(a,R \circ L(a)), \quad
    \epsilon : \prod_{b:\cat D} \hom(L\circ R (b), b),
  \end{displaymath}
  together with elements witnessing the naturality of the unit and the counit
  \begin{gather*}
    \nat_\eta : \prod_{a,a':\cat C} \prod_{f:\hom(a,a')} (R\circ L)(f) \circ \eta(a) = \eta(a') \circ f, \\
    \nat_\epsilon : \prod_{b,b':\cat D} \prod_{f:\hom(b,b')} f \circ \epsilon(b) = \epsilon(b') \circ (L\circ R)(f),
  \end{gather*}
  and elements witnessing the triangle identities
  \begin{align*}
    &\ltr_{\epsilon,\eta} : \prod_{a:\cat C} \epsilon(L(a)) \circ L(\eta(a)) = \id_{L(a)}, \\
    &\rtr_{\eta,\epsilon} : \prod_{b:\cat D} R(\epsilon(b)) \circ \eta(R(b)) = \id_{R(b)}.
  \end{align*}
\end{defi}
Here again, even though the types of $\nat_\eta$,
$\nat_\epsilon$, $\ltr_{\eta,\epsilon}$ and
$\rtr_{\eta,\epsilon}$ are not propositions, one only cares about their
existence, and therefore one usually omits them when denoting a wild
adjunctions.

\begin{rem} \label{def:wild-adj}
  As carefully proven and formalized in Agda~\cite[\githubpath
  theorems/homotopy/PtdAdjoint.agda]{hott-agda}, any such wild adjunction
  induces a dependent function $\Phi$ that maps elements $a:\cat C$ and $b:\cat
  D$ to an equivalence
  \begin{displaymath}
    \Phi_{a,b} : (\hom(L(a),b)) \weq (\hom(a,R(b)))
  \end{displaymath}
  given by $f \mapsto R(f) \circ \eta(a)$, with inverse the function given by
  $g\mapsto \epsilon(b) \circ L(g)$.

  This dependent function is natural in the following sense: for any $a,a':\cat
  C$ and $b,b':\cat D$, there are elements of $\Phi_{a,b}(\blank) \circ f =
  \Phi_{a',b} (\blank \circ L(f))$ for any $f:\hom(a',a)$ and
  $\Phi_{a,b'}(g\circ\blank) = R(g)\circ \Phi_{a,b} (\blank)$ for any
  $g:\hom(b,b')$.

  (To be precise, this is proven in \cite{hott-agda} only for wild functors $L$
  and $R$ from $\UUptd$ to $\UUptd$. This is the only case we need for this
paper, hence we rely on this proof.)
\end{rem}

\begin{prop}
  There is a wild adjunction $\susp \null \leftadjto \loopspace \null$.
  \label{prop:susp-loop-adjunction}
\end{prop}
\begin{proof}
  We refer to \cite[\githubpath
  theorems/homotopy/SuspAdjointLoop.agda]{hott-agda} for a proper proof.
  However, we give the unit $\eta$ and counit $\epsilon$ for convenience.

  Let $A$ be a pointed type with distinguished point $a_0:A$. Then define
  \begin{displaymath}
    \eta_A \defequi \eta(A) : A \ptdto \loopspace\null\susp A, \quad
    a \mapsto \inv{\mrd(a_0)} \cdot \mrd(a)
  \end{displaymath}
  which is pointed by path algebra.
  And define $\epsilon(A) : \susp \loopspace \null A \ptdto A$ by induction by
  setting
  \begin{displaymath}
    \epsilon(A)(\north) \defequi a_0,\quad
    \epsilon(A)(\south) \defequi a_0,\quad
    \ap{\epsilon(A)}\circ \mrd \defis \id_{\loopspace\null A}.\qedhere
  \end{displaymath}
\end{proof}

\begin{prop}
  Given any pointed type $B$, there is a wild adjunction
  ${\blank}*B \leftadjto (B \ptdto \loopspace\null\blank)$.
  \label{prop:join-loop-adjunction}
\end{prop}
\noindent Here, the right adjoint is the composition
of the loop functor $\loopspace\null$
with the covariant hom-functor from $\UUptd$ to itself.
\begin{proof}
  We refer to~\cite{joinloopadj} for a formalization in cubical Agda.
  Here we just give the unit $\eta$ and the counit $\epsilon$.

  Let $A$ be a pointed type with base point $a_0$. Then define
  \[
    \eta(A) : A \ptdto (B \ptdto \loopspace\null(A * B)),\quad
    a \mapsto (b \mapsto \tau_{a,b}),
  \]
  where $\tau_{a,b}$ is defined in~\eqref{eq:def-tau-ab}.
  And if $X$ is a pointed type with base point $x_0$, we define
  $\epsilon(X) : (B \ptdto \loopspace\null X) * B \ptdto X$
  by induction by setting,
  \[
    \epsilon(X)(\inl(f)) \defequi x_0,\quad
    \epsilon(X)(\inr(b)) \defequi x_0,\quad
    \ap{\epsilon(X)}(\glue(f,b)) \defis f(b).\qedhere
  \]
\end{proof}

\begin{rem}
  It is possible to obtain~\cref{prop:join-loop-adjunction}
  from a more general adjunction on pointed types involving the
  \emph{smash product}, ${\blank}\wedge B \leftadjto (B \ptdto \blank)$,
  where the smash product $A \wedge B$ is defined as the following
  pushout involving the wedge inclusion,
  \[
    \begin{tikzcd}
      A \vee B \ar[r,"i"]\ar[d]\ar[pushout] & A\times B \ar[d,"\inr"] \\
      \bn 1 \ar[r,"\inl"'] & A \wedge B.
    \end{tikzcd}
  \]
  This adjunction is described in details
  in~\cite[Sec.~4.3.3]{vandoorn:thesis}, and has been formalized in Lean.

  To get~\cref{prop:join-loop-adjunction},
  we'd also need the natural equivalence $A * B \simeq \susp(A\wedge B)$,
  and then we could calculate
  \begin{align*}
    (A * B \ptdto X) \simeq (\susp(A\wedge B) \ptdto X)
	&\simeq (A \wedge B \ptdto \loopspace\null X)\\
	&\simeq (A \ptdto (B \ptdto \loopspace\null X)).
  \end{align*}
  However, it is considerably less work to establish
  \cref{prop:join-loop-adjunction} directly,
  since both $\loopspace\null X$ and $B \ptdto \loopspace\null X$
  are \emph{purely homogeneous types},
  where a pointed type $(X,x_0)$ with base point $x_0:X$ is
  called purely homogeneous if $(X,x_0) =_{\UUptd} (X,x)$ for all $x:X$.
  Indeed, in this case, two pointed maps $(A,a_0) \ptdto (X,x_0)$
  are equal as long as the underlying maps $A\to X$ are.
  This follows from~\cref{lemma:homotopy-pointed-maps} below, using
  $e_x \defequi \loopspace\null(w(x))$ where
  $w$ is a witness of pure homogeneity.

  This notion of pure homogeneity is closely related to being an H-space.
  Indeed, \cref{lemma:homotopy-pointed-maps} also applies to left-invertible H-spaces,
  taking $e_x \defequi \loopspace\null(\mu(\blank,x))$,
  where $\mu : X\to X \to X$ is binary operation such that
  $\mu(\blank,x)$ is invertible for all $x:X$
  and satisfying $\mu(x_0,\blank) = \id_X$.

  The Hopf construction~\cite[Sec.~8.5.2]{HoTT} applies
  as well to any connected purely homogeneous type $(X,x_0)$,
  since also the maps $w(\blank,y)$ will be invertible thanks
  to $w(\blank,x_0)$ being homotopic to the identity.
  In particular, this means that we can only expect $\Sn n$ to be purely
  homogeneous for $n=0,1,3,7$.

  Note that any pointed connected type $(X,x_0)$ is \emph{merely homogeneous}
  in the sense that we have $\Trunc{(X,x_0)=_{\UUptd}(X,x)}$ for all $x:X$,
  as witnessed by the identity.
\end{rem}

\begin{lem}
  \label{lemma:homotopy-pointed-maps}%
 Let $(X,x_0)$ be a pointed type and
$e_x : (x_0 = x_0) \to (x=x)$ a family of equivalences parametrized by $x: X$.
Let $(A,a_0)$ also be a pointed type, $(f,f_0),(g,g_0)$ two pointed maps $(A,a_0) \ptdto (X,x_0)$,
and $h: f\sim g$ a homotopy. Then $(f,f_0) = (g,g_0)$.
\end{lem}
\begin{proof}
It suffices to give a  homotopy $h': f\sim g$ such that $h'(a_0) = \inv {g_0} f_0$.
Define $p \defequi \inv e_{g(a_0)}(\inv {g_0} f_0 \inv{h(a_0)}) : x_0=x_0$
and $h'(a) \defequi e_{g(a)}(p) h(a) : f(a)=g(a)$ for all $a: A$.
Then indeed $h': f\sim g$. Moreover,
$h'(a_0) \equiv e_{g(a_0)}(\inv e_{g(a_0)}(\inv {g_0} f_0 \inv{h(a_0)})) h(a_0) =  \inv {g_0} f_0$.
\end{proof}

\begin{rem}
  Define a \emph{wild monoid} to be a pointed type $(M,1)$ equipped with a
  function $\cdot : M \times M \to M$ and elements:
  \begin{gather*}
    \alpha: \prod_{x,y,z:M}x\cdot(y\cdot z) = (x\cdot y) \cdot z \\
    \iota: \prod_{x:M} (x\cdot 1 = x) \times (x = 1 \cdot x)
  \end{gather*}
  If $M$ is a set, then the types of $\alpha$ and $\iota$ become propositions
  and we are left with just a usual monoid. In particular, for any wild monoid
  $M$, the type $\setTrunc M$ has the induced monoid structure.

  Note that any wild category $\cat C$, for any object $a:\cat C$, induces a
  \emph{wild monoid} structure on $\hom(a,a)$, pointed at $\id_a$ with the
  multiplication $\circ(a,a,a)$, and the witnesses $\alpha(a,a,a,a)$ and $\iota(a,a)$
  coming from $\cat C$.

  For example, $\Sn n \ptdto \Sn n$ and $\Sn n \to \Sn n$, with composition of (pointed)
  maps, are wild monoids for each $n:\NN$. Set truncation
  yields ordinary  monoids $\setTrunc{\Sn n \ptdto \Sn n}$ and $\setTrunc{\Sn n \to \Sn n}$.
  \label{rem:wild-monoids}
\end{rem}

\end{document}